\RequirePackage{fix-cm}
\documentclass[smallcondensed]{svjour3}     
\smartqed  

\usepackage{graphicx}
\usepackage{natbib}
\usepackage{array}
\usepackage{float}
\usepackage{multirow}

\usepackage{amsthm}
\usepackage{amsmath}
\usepackage{amssymb}
\usepackage{stackrel}
\usepackage{graphicx}
\usepackage{setspace}
\usepackage{esint}
\usepackage{caption}
\providecommand{\tabularnewline}{\\}
\floatstyle{ruled}
\newfloat{algorithm}{tbp}{loa}
\providecommand{\algorithmname}{Algorithm}
\floatname{algorithm}{\protect\algorithmname}

\theoremstyle{plain}
\newtheorem{thm}{\protect\theoremname}[section]
\numberwithin{thm}{section}
\theoremstyle{remark}
\newtheorem{rem}{\protect\remarkname}[section]
\numberwithin{rem}{section}
\theoremstyle{plain}
\newtheorem{prop}{\protect\propositionname}[section]
\numberwithin{prop}{section}

\usepackage[labelfont=bf]{caption}
\usepackage{mathabx}
\usepackage{hyperref}

\makeatletter
\def\@seccntformat#1{\@ifundefined{#1@cntformat}%
   {\csname the#1\endcsname\quad}  
   {\csname #1@cntformat\endcsname}
}
\let\oldappendix\appendix 
\renewcommand\appendix{%
    \oldappendix
    \newcommand{\section@cntformat}{\appendixname~\thesection:\;}
}
\makeatother

\begin{document}

\title{An algorithm for approximating the second moment of the normalizing
constant estimate from a particle filter
}

\titlerunning{Pairs algorithm}        

\author{Svetoslav Kostov         \and
        Nick Whiteley 
}

\authorrunning{S.Kostov \and N.Whiteley} 

\institute{Svetoslav Kostov \at
              School of Mathematics, University Walk, Bristol, BS8 1TW, UK \\
              \email{svetoslav.kostov@bristol.ac.uk}           
           \and
           Nick Whiteley \at
              School of Mathematics, University Walk, Bristol, BS8 1TW, UK \\
              \email{nick.whiteley@bristol.ac.uk}
}

\date{Received: date / Accepted: date}

\maketitle

\begin{abstract}
We propose a new algorithm for approximating the non-asymptotic second
moment of the marginal likelihood estimate, or normalizing constant,
provided by a particle filter. The computational cost of the new method
is $O(M)$ per time step, independently of the number of particles
$N$ in the particle filter, where $M$ is a parameter controlling
the quality of the approximation. This is in contrast to $O(MN)$
for a simple averaging technique using $M$ i.i.d. replicates of a
particle filter with $N$ particles. We establish that the approximation
delivered by the new algorithm is unbiased, strongly consistent and,
under standard regularity conditions, increasing $M$ linearly with
time is sufficient to prevent growth of the relative variance of the
approximation, whereas for the simple averaging technique it can be
necessary to increase $M$ exponentially with time in order to achieve
the same effect. This makes the new algorithm useful as part of strategies for estimating Monte Carlo variance. Numerical examples illustrate performance in the
context of a stochastic Lotka\textendash Volterra system and a simple
AR(1) model.
\keywords{marginal likelihood \and normalizing constant \and hidden Markov model \and particle
filter}
\end{abstract}

\section{Introduction}
\label{sec:SMC-and-HMM}

Particle filters, also known as Sequential Monte Carlo (SMC) methods
\citep{DoucetSMCMethodsInPractice}, are used across a variety of
disciplines including systems biology, econometrics, neuroscience
and signal processing, to perform approximate inferential calculations
in general state-space Hidden Markov Models (HMM) and in particular,
provide an unbiased estimate of the marginal likelihood. Recent application
areas of these techniques include for example, systems biology \citep{golightly2011bayesian,golightly2015delayed},
where the calculation of the marginal likelihood (ML) plays an important
role in the estimation of the parameters of stochastic models of biochemical
networks. Estimation of the marginal likelihood also features centrally
in Particle Markov Chain Monte Carlo methods \citep{AndrieuPMCMC}.

In the present paper we address the problem of approximating the non-asymptotic
second moment of the particle filter estimate of the marginal likelihood, henceforth for brevity ``the second moment''. As part of strategies to estimate Monte Carlo variance, this allows one to report a numerical measure of the reliability of the particle filter estimate.  Our
contributions are to introduce a new particle ``Pairs algorithm''
and prove that it unbiasedly and consistently approximates the second
moment. We also establish, under regularity conditions, a linear-in-time
bound on the relative variance of the approximation to the second
moment, and illustrate through a simple calculation and numerical
simulations, that the Pairs algorithm performs more reliably than
a default strategy which uses independent copies of the particle filter.
In order to discuss the connections between our work and the existing
literature, we first need to introduce some notation and definitions.

A HMM is a process $(X_{n},Y_{n})_{n\geq0}$, where $\left(X_{n}\right)_{n\geq0}$,
called the signal process, is a Markov chain with state space $\mathsf{X}$,
initial distribution $\pi_{0}$ and transition kernel $f$. Each of
the observations $Y_{n}\in\mathsf{Y}$, is conditionally independent
of the rest of the signal process given $X_{n}$, with conditional distribution,
$g(X_{n},\cdot)$, where $g$ is a probability kernel from $\mathsf{X}$
to $\mathsf{Y}$. The HMM can be represented as:
\begin{eqnarray}
X_{0}\sim\pi_{0}(\cdot), &  & X_{n}\mid X_{n-1}\sim f(X_{n-1},\cdot),\quad n\geq1\label{eq:HMM-def}\\
 &  & \; Y_{n}\mid X_{n}\sim g(X_{n},\cdot),\quad n\geq0.\nonumber
\end{eqnarray}

We consider a fixed observation sequence $(y_{n})_{n\geq0}$, assume
that $g$ admits a density $g(x,y)$ w.r.t. to some dominating measure
and write for brevity $g_{n}(x)=g(x,y_{n})$. For simplicity we also
assume throughout that for all $n\geq0$, $\sup_{x}g_{n}(x)<+\infty$
and $g_{n}\left(x\right)>0$, $\forall x\in\mathsf{X}$. We then define
the sequence of distributions $(\pi_{n})_{n\geq1}$, called prediction
filters, as
\[
\pi_{n+1}(A):=\dfrac{\int_{\mathsf{X}}\pi_{n}(dx)g_{n}(x)f(x,A)}{\int_{\mathsf{X}}\pi_{n}(dx)g_{n}(x)},\quad\forall A\in\mathcal{X},\; n\geq0,
\]
where $\mathcal{X}$ is the $\sigma$-algebra associated with the
space $\mathsf{X}$, and the sequence
\begin{equation}
(Z_{n})_{n\geq0},\quad Z_{0}:=\int_{\mathsf{X}}g_{0}(x)\pi_{0}(dx),\quad Z_{n}:=Z_{n-1}\int_{\mathsf{X}}g_{n}(x)\pi_{n}(dx),\quad n\geq1.\label{eq:Z_n_defn}
\end{equation}

The interpretation of these definitions is the following: $\pi_{n+1}$
is the distribution of $X_{n+1}\mid Y_{0:n}=y_{0:n}$, where for any
sequence $(a_{n})_{n\geq0}$ we write $a_{p:q}=(a_{p},\ldots,a_{q})$,
and $Z_{n}$ is the marginal likelihood of the first $n+1$ observations
$y_{0:n}$. In many cases of interest, the distributions $\pi_{n}$
and constants $Z_{n}$ cannot be computed exactly, and numerical approximations
are needed. A particle filter, shown in Algorithm \ref{alg:standard-SMC-algorithm},
provides such approximations, denoted respectively $\pi_{n}^{N}$
and $Z_{n}^{N}$. In Algorithm \ref{alg:standard-SMC-algorithm} $q_{0}$
and $q_{n}$, $n\geq1$ are respectively a distribution and Markov
kernels on $\mathsf{X}$, which may depend on the observations sequence
$(y_{n})_{n\geq0}$, but this dependence is suppressed from the notation.
We assume throughout the rest of the paper that $\pi_{0}(\cdot)$,
$f(x,\cdot)$ and $q_{0}(\cdot)$ and $q_{n}(x,\cdot)$ admit a density
w.r.t. to some common dominating measure $dx$, and with a slight
abuse of notation, the corresponding densities are denoted by $\pi_{0}(x)$,
$f(x,x^{\prime})$, $q_{0}(x)$ and $q_{n}(x,x^{\prime})$.
\begin{algorithm}
\protect\caption{\label{alg:standard-SMC-algorithm}SMC algorithm for estimating $Z_{n}$
using $N$ particles}

Initialization
\begin{itemize}
\item Sample $\left\{ X_{0}^{i}\right\} _{i=1}^{N}\overset{i.i.d.}{\sim}q_{0}(\cdot)$
\item Compute weights $\left\{ W_{0}^{i}\right\} _{i=1}^{N}$ according
to $W_{0}^{i}=\dfrac{g_{0}(X_{0}^{i})\pi_{0}(X_{0}^{i})}{q_{0}(X_{0}^{i})}$

normalize, $\widetilde{W}_{0}^{i}=\dfrac{W_{0}^{i}}{\sum_{k=1}^{N}W_{0}^{k}}$,
and set $Z_{0}^{N}=\dfrac{1}{N}\sum_{i=1}^{N}W_{0}^{i}$

\item Resample conditionally i.i.d. draws from $\left\{ X_{0}^{i}\right\} _{i=1}^{N}$
using the normalized weights $\left\{ \widetilde{W}_{0}^{i}\right\} _{i=1}^{N}$
to obtain a set of equally-weighted particles $\left\{ \underline{X}_{0}^{i}\right\} _{i=1}^{N}$
\end{itemize}
For $n\geq1$:
\begin{itemize}
\item For each $i$, set $X_{n-1}^{i}=\underline{X^{i}}_{n-1}$
\item For each $i$, sample $X_{n}^{i}\sim q_{n}(X_{n-1}^{i},\cdot)$, compute
weights $W_{n}^{i}=\dfrac{g_{n}(X_{n}^{i})f(X_{n-1}^{i},X_{n}^{i})}{q_{n}(X_{n-1}^{i},X_{n}^{i})}$,

normalize, $\widetilde{W}_{n}^{i}=\dfrac{W_{n}^{i}}{\sum_{k=1}^{N}W_{n}^{k}}$,
and set $Z_{n}^{N}=Z_{n-1}^{N}\cdot\left(\dfrac{1}{N}\sum_{i=1}^{N}W_{n}^{i}\right)$

\item Resample conditionally i.i.d. draws from $\left\{ X_{n}^{i}\right\} _{i=1}^{N}$
using the normalized weights $\left\{ \widetilde{W}_{n}^{i}\right\} _{i=1}^{N}$
to obtain a set of equally-weighted particles $\left\{ \underline{X^{i}}_{n}\right\} _{i=1}^{N}$\end{itemize}
\end{algorithm}

It is well known that Algorithm \ref{alg:standard-SMC-algorithm}
provides an unbiased estimate of $Z_{n}$, i.e. $\mathbb{E}\left[Z_{n}^{N}\right]=Z_{n}$.
A detailed account of this fact is given in \citep[Ch. 9]{DelMoral04}.
The main contribution of the present paper is to propose and study
a new method to approximate $\mathbb{E}\left[\left(Z_{n}^{N}\right)^{2}\right]$.
The approximation is delivered by Algorithm \ref{alg:PairsAlgo} -- the Pairs algorithm --
which we introduce in the next section, and which must be run \emph{in
addition} to the particle filter used to estimate $Z_{n}^{N}$.  Our main motivation for approximating $\mathbb{E}\left[\left(Z_{n}^{N}\right)^{2}\right]$ is to calculate $\text{Var}\left[Z_{n}^{N}\right]$. In
a recent arXiv manuscript \citep{lee2015variance}, A. Lee and the
second author of the present paper have introduced a method which
allows one to unbiasedly approximate $\text{Var}\left[Z_{n}^{N}\right]$
using the same single run of the particle filter which delivers $Z_{n}^{N}$.
As $N\to\infty$, the method of \citet{lee2015variance} allows one
to consistently approximate asymptotic variance $\lim_{N\to\infty}N\text{Var}\left[Z_{n}^{N}\right]$.

We stress that the Pairs algorithm performs the different
task of approximating, for any \emph{fixed} $N\geq2$, the non-asymptotic
quantity $\mathbb{E}\left[\left(Z_{n}^{N}\right)^{2}\right]$ to arbitrary
accuracy controlled by an auxiliary parameter $M$ (this statement
is made precise in Theorem \ref{thm:Pairs-convergence-and-bound}
below). Thus the Pairs algorithm allows one to reliably approximate $\mathbb{E}\left[\left(Z_{n}^{N}\right)^{2}\right]$ without requiring that $N$ is large. We shall later illustrate how this property makes the Pairs algorithm useful within strategies for estimating $\text{Var}\left[Z_{n}^{N}\right]$.

Moreover in Theorem \ref{thm:Pairs-convergence-and-bound} we prove an important
result regarding the time dependence of the error of the approximation of $\mathbb{E}\left[\left(Z_{n}^{N}\right)^{2}\right]$ delivered by the Pairs algorithm, showing that under standard regularity conditions, it is sufficient to increase $M$ linearly with $n$ to control the relative variance of this approximation.
This is in contrast to \citet{lee2015variance}, who do not provide any results concerning the time-dependence of the errors associated with their estimators. 

We note that \citet{chan2013} investigated numerical techniques for assessing
the asymptotic variance associated with particle estimates of expectations
with respect to filtering distributions, but they didn't explore methods
for approximating $\mathbb{E}\left[\left(Z_{n}^{N}\right)^{2}\right]$. We also note that \citet{bhadra2014adaptive}
proposed to approximate $\mathbb{E}\left[\left(Z_{n}^{N}\right)^{2}\right]$
using a ``meta-model'', for purposes of optimizing parameters of
the particle filter. Their method amounts to fitting an AR(1) process
to the output of the particle filter; it seems difficult to assess
the bias of their approach and no proof of consistency is given.

\section{Pairs algorithm}
\label{sec:Pairs-algorithm}

\subsection{Outline of how the algorithm is derived}

The full details of the derivation of the Pairs algorithm are given in Appendix \ref{sec:Appendix}. We now give an account of some of the main ideas behind this derivation. For this some more notation is needed. Let us introduce the nonnegative integral kernels:
for $x\in\mathsf{X},y=(y_{1},y_{2})\in\mathsf{X}^{2}$,
\begin{equation}
Q_{1}(x,dy) = \frac{g_{0}(x)\pi_{0}(x)}{q_{0}(x)}q_{1}(y_{1},y_{2})
\delta_{x}(dy_{1})dy_{2},
\end{equation}
and for $n\geq2$ and $x=(x_{1},x_{2}) \in\mathsf{X}^{2},y\in\mathsf{X}^{2}$,
\begin{equation}
Q_{n}(x,dy) = \frac{g_{n-1}(x_{2})f(x_{1},x_{2})}{q_{n-1}(x_{1},x_{2})}
q_{n}(y_{1},y_{2})\delta_{x_{2}}(dy_{1})dy_{2}.
\end{equation}
In terms of compositions of these kernels, the lack-of-bias property of the particle filter reads as:
\begin{equation}
\mathbb{E}\left[Z_n^N\right]=\pi_0Q_1\cdots Q_n(1).\label{eq:unbiased_Q}
\end{equation}
The kernels also encapsulate the main ingredients of the particle filter itself, indeed one may take the point of view that Algorithm \ref{alg:standard-SMC-algorithm} is actually derived from the $Q_n$, in the sense that resampling is performed according to weights given by evaluating the functions
\begin{equation}
Q_{1}(x,\mathsf{X}^2)=\frac{g_{0}(x)\pi_{0}(x)}{q_{0}(x)},\quad Q_n(x,\mathsf{X}^2)=\frac{g_{n-1}(x_{2})f(x_{1},x_{2})}{q_{n-1}(x_{1},x_{2})},\;\;n\geq2,\label{eq:Q_G_front}
\end{equation}
and sampling is performed using the the Markov kernels:
\begin{equation}
\frac{Q_{n}(x,\cdot)}{Q_{n}(x,\mathsf{X}^2)}.\label{eq:Q_M_front}
\end{equation}

Now introduce the so--called coalescence operator $C$ which acts on
functions $F:\mathsf{X}^{2}\times \mathsf{X}^{2} \rightarrow \mathbb{R}$ as $C(F)(x,y)=F(x,x)$. \citet{CerouDelMoral2011} derived the following representation
of the second moment of $Z_{n}^{N}$,
\begin{equation}
\mathbb{E}\left[\left(Z_{n}^{N}\right)^{2}\right]=\mathbb{E}\left[\pi_{0}^{\otimes2}C_{\epsilon_{0}}Q_{1}^{\otimes2}C_{\epsilon_{1}}\cdots C_{\epsilon_{n}}Q_{n+1}^{\otimes2}(1)\right],
\label{eq:gamma-squared-representation_intro}
\end{equation}
where $C_{1}:=C$, $C_{0}:=Id$, $\left\{ \epsilon_{n}\right\} _{n\geq0}$ is a sequence
of i.i.d., $\left\{ 0,1\right\}$-valued random variables
with distribution
\[
\mathbb{\mathbb{P}}(\epsilon_{n}=1)=1-\mathbb{P}(\epsilon_{n}=0)=\frac{1}{N},
\]
and $Q_{n}^{\otimes2}$ is the two-fold tensor product of $Q_n$.

The main idea behind the Pairs algorithm is to identify, using \eqref{eq:gamma-squared-representation_intro}, certain nonnegative kernels $\mathbf{Q}_n ^{(N)}$ such that the second moment can be written
$$
\mathbb{E}\left[\left(Z_{n}^{N}\right)^{2}\right]=\pi_0^{\otimes2} \mathbf{Q}_1^{(N)}\cdots \mathbf{Q}_{n+1} ^{(N)}(1).
$$
The details of these kernels $\mathbf{Q}_n ^{(N)}$ are given in the Appendix. Observing the similarity with \eqref{eq:unbiased_Q}, to obtain the Pairs algorithm we shall derive a particle algorithm from the weighting functions and Markov kernels which are associated with $\mathbf{Q}_n^{(N)}$ in the same way as \eqref{eq:Q_G_front}-\eqref{eq:Q_M_front} are associated with $Q_n$, the result being the Pairs algorithm. Results for standard particle filters then transfer to the Pairs algorithm directly, which leads to our Theorem \ref{thm:Pairs-convergence-and-bound} below.

%
%

\subsection{The algorithm and its properties}

In Algorithm \ref{alg:PairsAlgo} both $N\geq2$ and $M\geq1$ are
parameters. The computational cost of Algorithm \ref{alg:PairsAlgo}
is $O(M)$ per time step, uniformly in $N$, and the quantity $\Xi_{n}^{\left(N,M\right)}$
which it delivers can be considered an approximation to $\mathbb{E}\left[\left(Z_{n}^{N}\right)^{2}\right]$,
in the sense of Theorem \ref{thm:Pairs-convergence-and-bound} below.

\begin{algorithm}
\protect\caption{\label{alg:PairsAlgo}Pairs algorithm for approximating $\mathbb{E}\left[\left(Z_{n}^{N}\right)^{2}\right]$
using $M$ pair particles}

Initialization
\begin{itemize}
\item Sample pairs $\left\{ \check{X}_{0}^{i}\right\} _{i=1}^{M}\overset{i.i.d.}{\sim}q_{0}(\cdot)$,
$\left\{ \hat{X}_{0}^{i}\right\} _{i=1}^{M}\overset{i.i.d.}{\sim}q_{0}(\cdot)$
\item Compute weights $\left\{ W_{0}^{i}\right\} _{i=1}^{M}$ according
to
\[
W_{0}^{i}=\dfrac{1}{N}\dfrac{g_{0}(\check{X}_{0}^{i})^{2}\pi_{0}(\check{X}_{0}^{i})^{2}}{q_{0}(\check{X}_{0}^{i})^{2}}+\left(1-\dfrac{1}{N}\right)\dfrac{g_{0}(\check{X}_{0}^{i})g_{0}(\hat{X}_{0}^{i})\pi_{0}(\check{X}_{0}^{i})\pi_{0}(\hat{X}_{0}^{i})}{q_{0}(\check{X}_{0}^{i})q_{0}(\hat{X}_{0}^{i})},
\]

normalize weights $\widetilde{W}_{0}^{i}=\dfrac{W_{0}^{i}}{\sum_{k=1}^{M}W_{0}^{k}}$
and set $\Xi_{0}^{\left(N,M\right)}=\dfrac{1}{M}\sum_{i=1}^{M}W_{0}^{i}$.

\item Resample conditionally i.i.d. draws from $\left\{ \check{X}_{0}^{i},\hat{X}_{0}^{i}\right\} _{i=1}^{M}$
using the normalized weights $\left\{ \widetilde{W}_{0}^{i}\right\} _{i=1}^{M}$
to obtain a set of equally-weighted particles $\left\{ \check{\underline{X}}_{0}^{i},\hat{\underline{X}}_{0}^{i}\right\} _{i=1}^{M}$
\item For each $i$, set $\left(\check{X}_{0}^{i},\hat{X}_{0}^{i}\right)=\left(\check{\underline{X}}_{0}^{i},\hat{\underline{X}}_{0}^{i}\right)$,
compute $p_{0}^{i}=\left(1+(N-1)\dfrac{g_{0}(\hat{X}_{0}^{i})\pi_{0}(\hat{X}_{0}^{i})q_{0}(\check{X}_{0}^{i})}{g_{0}(\check{X}_{0}^{i})\pi_{0}(\check{X}_{0}^{i})q_{0}(\hat{X}_{0}^{i})}\right)^{-1}$
and sample $Y_{0}^{i}\sim Ber(p_{0}^{i})$. If $Y_{0}^{i}=1$, set
$\hat{X}_{0}^{i}=\check{X}_{0}^{i}$. Sample $\check{X}_{1}^{i}\sim q_{1}(\check{X}_{0}^{i},\cdot)$,
$\hat{X}_{1}^{i}\sim q_{1}(\hat{X}_{0}^{i},\cdot)$.
\end{itemize}
For $n\geq1$:
\begin{itemize}
\item Compute weights $\left\{ W_{n}^{i}\right\} _{i=1}^{M}$ according
to
\begin{eqnarray*}
W_{n}^{i} & = & \dfrac{1}{N}\dfrac{g_{n}(\check{X}_{n}^{i})^{2}f(\check{X}_{n-1}^{i},\check{X}_{n}^{i})^{2}}{q_{n}(\check{X}_{n-1}^{i},\check{X}_{n}^{i})^{2}}\\
 & + & \left(1-\dfrac{1}{N}\right)\dfrac{g_{n}(\check{X}_{n}^{i})g_{n}(\hat{X}_{n}^{i})f(\check{X}_{n-1}^{i},\check{X}_{n}^{i})f(\hat{X}_{n-1}^{i},\hat{X}_{n}^{i})}{q_{n}(\check{X}_{n-1}^{i},\check{X}_{n}^{i})q_{n}(\hat{X}_{n-1}^{i},\hat{X}_{n}^{i})},
\end{eqnarray*}

normalize, $\widetilde{W}_{n}^{i}=\dfrac{W_{n}^{i}}{\sum_{k=1}^{M}W_{n}^{k}},$
and set $\Xi_{n}^{\left(N,M\right)}=\Xi_{n-1}^{\left(N,M\right)}\cdot\left(\dfrac{1}{M}\sum_{i=1}^{M}W_{n}^{i}\right)$

\item Resample conditionally i.i.d. draws from $\left\{ \check{X}_{n-1:n}^{i},\hat{X}_{n-1:n}^{i}\right\} _{i=1}^{M}$
using the normalized weights $\left\{ \widetilde{W}_{n}^{i}\right\} _{i=1}^{M}$
to obtain a set of equally-weighted particles $\left\{ \check{\underline{X}}_{n-1:n}^{i},\hat{\underline{X}}_{n-1:n}^{i}\right\} _{i=1}^{M}$
\item For each $i$, set $\left(\check{X}_{n-1:n}^{i},\hat{X}_{n-1:n}^{i}\right)=\left(\check{\underline{X}}_{n-1:n}^{i},\hat{\underline{X}}_{n-1:n}^{i}\right)$,
compute $p_{n}^{i}=\left(1+(N-1)\dfrac{g_{n}(\hat{X}_{n}^{i})f(\hat{X}_{n-1}^{i},\hat{X}_{n}^{i})q_{n}(\check{X}_{n-1}^{i},\check{X}_{n}^{i})}{g_{n}(\check{X}_{n}^{i})f(\check{X}_{n-1}^{i},\check{X}_{n}^{i})q_{n}(\hat{X}_{n-1}^{i},\hat{X}_{n}^{i})}\right)^{-1}$
and sample $Y_{n}^{i}\sim Ber(p_{n}^{i})$. If $Y_{n}^{i}=1$, set
$\hat{X}_{n}^{i}=\check{X}_{n}^{i}$. Sample $\check{X}_{n+1}^{i}\sim q_{n+1}(\check{X}_{n}^{i},\cdot)$,
$\hat{X}_{n+1}^{i}\sim q_{n+1}(\hat{X}_{n}^{i},\cdot)$.\end{itemize}
\end{algorithm}

\begin{thm}
\label{thm:Pairs-convergence-and-bound}If
\begin{equation}
\sup_{x}\dfrac{g_{0}(x)\pi_{0}(x)}{q_{0}(x)}<+\infty\quad\text{ and }\quad\sup_{x_{1},x_{2}}\dfrac{g_{n}(x_{2})f(x_{1},x_{2})}{q_{n}(x_{1},x_{2})}<+\infty,\quad\forall n\geq1,\label{eq:main_thm_weights_bounded}
\end{equation}
then for any $N\geq2$ and $n\geq0$,
\begin{eqnarray*}
\mathbb{E}\left[\Xi_{n}^{\left(N,M\right)}\right] & = & \mathbb{E}\left[\left(Z_{n}^{N}\right)^{2}\right],\quad\forall M\geq1,\\
\Xi_{n}^{\left(N,M\right)} & \stackrel[M\rightarrow\infty]{a.s.}{\longrightarrow} & \mathbb{E}\left[\left(Z_{n}^{N}\right)^{2}\right].
\end{eqnarray*}
If additionally for each $n\geq0$ there exist constants $0<w_{n}^{-}\leq w_{n}^{+}<+\infty$,
and for each $n\geq1$, constants $0<\epsilon_{n}^{-}\leq\epsilon_{n}^{+}<+\infty$
and a probability measure $\mu_{n}$ such that
\begin{eqnarray}
 &  & w_{0}^{-}\leq g_{0}(x)\pi_{0}(x)/q_{0}(x)\leq w_{0}^{+},\quad\forall x,\label{eq:main_thm_reg_weights}\\
 &  & w_{n}^{-}\leq g_{n}(x_{2})f(x_{1},x_{2})/q_{n}(x_{1},x_{2})\leq w_{n}^{+},\quad\forall x_{1},x_{2},n\geq1,\label{eq:main_thm_reg_weights-1}\\
 &  & \epsilon_{n}^{-}\mu_{n}(\cdot)\leq q_{n}(x,\cdot)\leq\epsilon_{n}^{+}\mu_{n}(\cdot),\quad\forall x,n\geq1,\label{eq:main_thm_reg_q}
\end{eqnarray}
then for any $N\geq2$ and $n\geq0$,
\[
M>\sum_{s=0}^{n+1}\Delta_{s}\quad\Rightarrow\quad\mathbb{E}\left[\left(\frac{\Xi_{n}^{(N,M)}}{\mathbb{E}\left[\left(Z_{n}^{N}\right)^{2}\right]}-1\right)^{2}\right]\leq\frac{4}{M}\sum_{s=0}^{n+1}\Delta_{s}
\]
where $\Delta_{s}:=\left(\frac{w_{s}^{+}w_{s+1}^{+}\epsilon_{s+1}^{+}}{w_{s}^{-}w_{s+1}^{-}\epsilon_{s+1}^{-}}\right)^{2}$
is independent of $M$ and $N$.
\end{thm}
The proof of Theorem \ref{thm:Pairs-convergence-and-bound} is given
in Appendix \ref{sec:Appendix}. The conditions in (\ref{eq:main_thm_reg_weights})-(\ref{eq:main_thm_reg_q})
are fairly standard in the stability theory of particle filters, but
are rather strong: they rarely hold when $\mathsf{X}$ is an unbounded
subset of $\mathbb{R}^{d}.$ Attempting to establish similar results
under more realistic conditions, for example via the techniques of
\citet{whiteley2013stability}, seems to be a much more difficult
task, beyond the scope of the present work, and we leave a full investigation
of this matter to future research.

\subsection{Comparison to using i.i.d. replicates of $Z_{n}^{N}$}

A natural alternative to $\Xi_{n}^{(N,M)}$ as an approximation to
$\mathbb{E}\left[\left(Z_{n}^{N}\right)^{2}\right]$ is to use $M$
i.i.d. replicates $\left\{ Z_{n}^{N,j}\right\} _{j=1}^{M}$ of $Z_{n}^{N}$
and simple averaging,
\begin{equation}
\widetilde{\Xi}_{n}^{(N,M)}:=\frac{1}{M}\sum_{j=1}^{M}\left(Z_{n}^{N,j}\right)^{2}.\label{eq:Xi_tilde}
\end{equation}
The cost of computing $\widetilde{\Xi}_{n}^{(N,M)}$ is $O(MN)$ per
time step since it involves $M$ copies of Algorithm \ref{alg:standard-SMC-algorithm},
each using $N$ particles.

To illustrate why $\Xi_{n}^{(N,M)}$ is to be preferred over $\widetilde{\Xi}_{n}^{(N,M)}$
in terms of relative variance, consider for simplicity of exposition
the case: for $n\geq1$, $q_{n}(x,\cdot)=f(x,\cdot)=\pi_0 (\cdot)$; for
$n=0$, $q_{0}(\cdot)=\pi_{0}(\cdot)$; and for $n\geq0$, $g_{n}(x)=g(x)$. In this case, for all $n\geq0$, we have $\pi_{n}=\pi_{0}$ and in Algorithm \ref{alg:standard-SMC-algorithm},
$\{X_{n}^{i}\}_{i=1}^{N}$ are i.i.d. draws from $\pi_{0}$. Then
with $\pi_{p}^{N}(g):=N^{-1}\sum_{i=1}^{N}g(X_{p}^{i})$, $Z_{n}^{N}=\prod_{p=0}^{n}\pi_{p}^{N}(g)$,
and 
\begin{eqnarray}
\mathbb{E}\left[\left(\frac{\widetilde{\Xi}_{n}^{(N,M)}}{\mathbb{E}\left[\left(Z_{n}^{N}\right)^{2}\right]}-1\right)^{2}\right] & = & \frac{1}{M}\left(\frac{\mathbb{E}\left[\left(Z_{n}^{N}\right)^{4}\right]}{\mathbb{E}\left[\left(Z_{n}^{N}\right)^{2}\right]^{2}}-1\right)\nonumber \\
 & = & \frac{1}{M}\left(\prod_{p=0}^{n}\frac{\mathbb{E}\left[\pi_{p}^{N}(g)^{4}\right]}{\mathbb{E}\left[\pi_{p}^{N}(g)^{2}\right]^{2}}-1\right)\nonumber \\
 & = & \frac{1}{M}\left(C^{n+1}-1\right),\label{eq:Xi_tilde_bound}
\end{eqnarray}
where $C:=\mathbb{E}\left[\pi_{0}^{N}(g)^{4}\right]/\mathbb{E}\left[\pi_{0}^{N}(g)^{2}\right]^{2}\geq1$
by Jensen's inequality, with equality holding if and only if $\pi_{0}^{N}(g)$
is a.s. constant. So if $\pi_{0}^{N}(g)$ exhibits any stochastic
variability at all, in the sense that $C>1$, then $M$ must be scaled
exponentially fast with $n$ in order to control (\ref{eq:Xi_tilde_bound}),
cf. the linear-in-$n$ scaling in Theorem \ref{thm:Pairs-convergence-and-bound}.

\section{Numerical examples}\label{sec:examples}

We will illustrate the properties of the Pairs algorithm using two
numerical examples. The first, in Section \ref{sub:ar_1} is a simple toy example, based
on a $AR(1)$ auto-regressive process. The
second, in Section \ref{sub:LV}, is a more realistic example involving a Lotka - Volterra
system of ODEs, observed in noise. In Section \ref{sub:Strategies_Comparison} we investigated the performance of the pairs algorithm within a strategy for estimating Monte Carlo variance.

Throughout section \ref{sec:examples} we denote by $M^{\prime}$ a number of pairs used in the
pairs algorithm to obtain a reliable, benchmark estimate of the true quantity
$\mathbb{E}\left[\left(Z_{n}^{N}\right)^{2}\right]$.


\subsection{$AR(1)$ example}\label{sub:ar_1}


The signal of this model $\left(X_{n}\right)_{n\geq0}$ is an $AR(1)$
process, defined by $X_{n+1}=\alpha X_{n}+\epsilon_{n+1}$, where
we set $\alpha=0.5$, $\epsilon_{n}\sim\mathcal{N}(0,\sigma^{2})$,
$\sigma=10$. Assume that $g_{n}(x)=\exp\left(-x^{2}/100\right),\forall n$.
We will also assume that $q_{n}(x,\cdot)=f(x,\cdot)$, i.e. we will
propose using the actual signal density and we will set $q_{0}=\pi_{0}$, given by $X_{0}\sim\mathcal{N}(0,\sigma^{2}/(1-\alpha^{2}))$,
i.e. the process $\left(X_{n}\right)_{n\geq0}$ is stationary a priori.

In Figure \ref{fig:ar1-example} we compare two approaches for
estimating $\mathbb{E}\left[\left(Z_{n}^{N}\right)^{2}\right]$: using the Pairs algorithm, and the standard MC approach using i.i.d. replicates as in \eqref{eq:Xi_tilde}. We consider two sub--examples: the first one
is for comparatively small number of particles $N=50$, and the second
sub--example is with higher number of particles $N=250$. The plots show $\log(\Xi_{n}^{\left(N,M\right)})-\log(\Xi_{n}^{(N,M^{\prime})})$
for the Pairs algorithm and $\log(\widetilde{\Xi}_{n}^{(N,\tilde{M})})-\log(\Xi_{n}^{(N,M^{\prime})})$
for the standard MC approach (please refer to Algorithm \ref{alg:PairsAlgo} and
(\ref{eq:Xi_tilde})). Here we take $M^{\prime}=10^{6}$
so that $\Xi_{n}^{(N,M^{\prime})}$ is a reliable, benchmark value of
$\mathbb{E}\left[\left(Z_{n}^{N}\right)^{2}\right]$.

In the top left plot of Figure \ref{fig:ar1-example} we
have chosen $M=\tilde{M}=10^{4}$. For the equal cost plot on the
top right we have chosen $M=10^{4}$ and $\tilde{M}=2500$. Here,
by ``equal cost'' we mean  that $M$ and $\tilde{M}$ are chosen such that the execution times of the standard MC
algorithm and the Pairs algorithm are the same. The time parameter
$n$ varies from $0$ to $500$ in both plots and we plot $20$ independent
runs of both algorithms in order to compare their variability properties.

The second row of plots in Figure \ref{fig:ar1-example} consists
of plots for the case of larger number of particles $N=250$. Again,
in the bottom left we are comparing the case where $M=\tilde{M}=10^{4}$,
and in bottom right we are comparing the equal cost case where $M=10^{4}$
and $\tilde{M}=700$. The fact that $\tilde{M}$ is lower here than in the $N=50$ case reflects the fact that the cost of the standard MC approach is $O(\tilde{M}N)$ per time step, compared to $O(M)$ for the Pairs algorithm. We have plotted 20 independent runs for both
algorithms. 

\begin{figure*}
\begin{centering}
\includegraphics[scale=0.4]{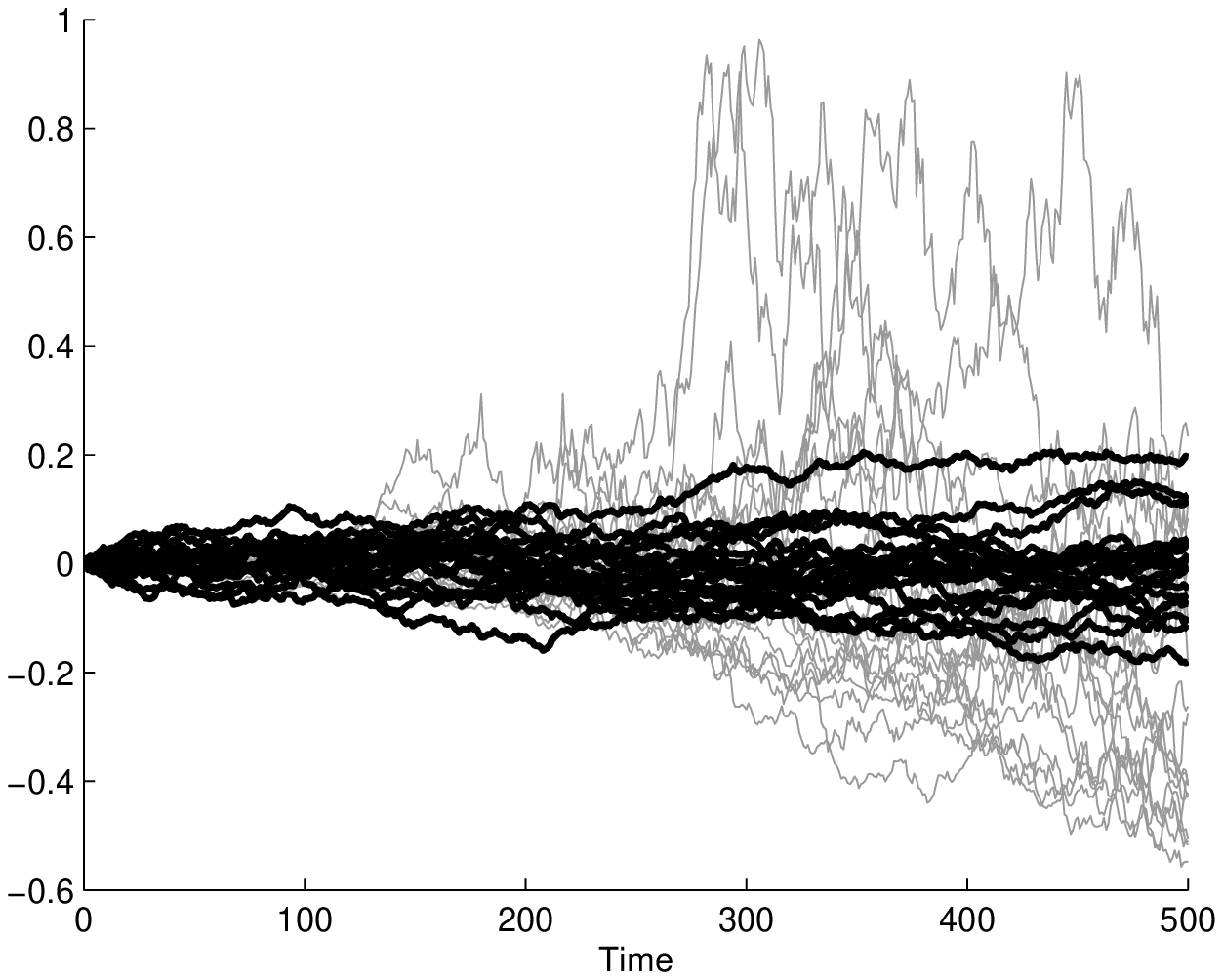}\includegraphics[scale=0.4]{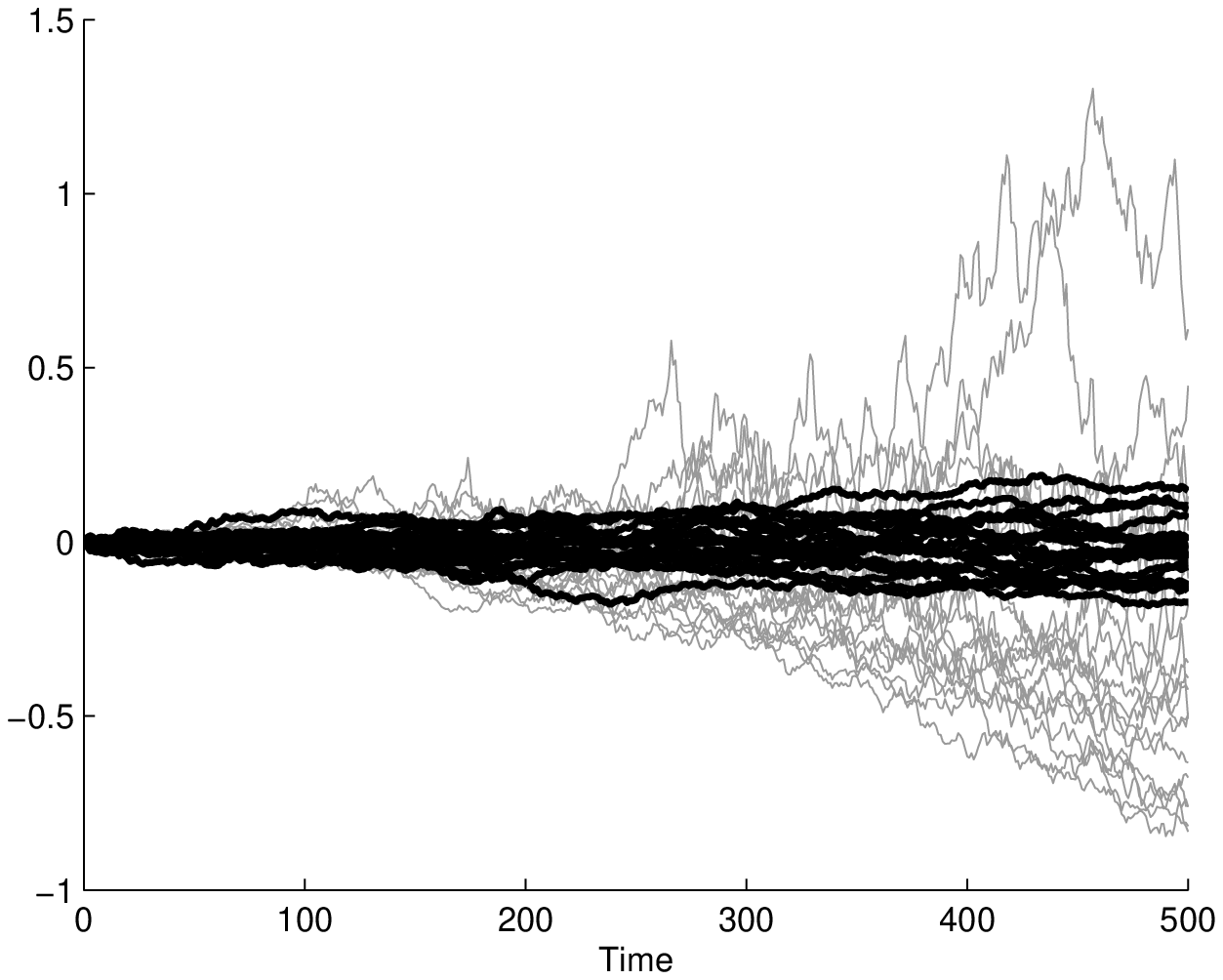}
\par\end{centering}

\begin{centering}
\includegraphics[scale=0.4]{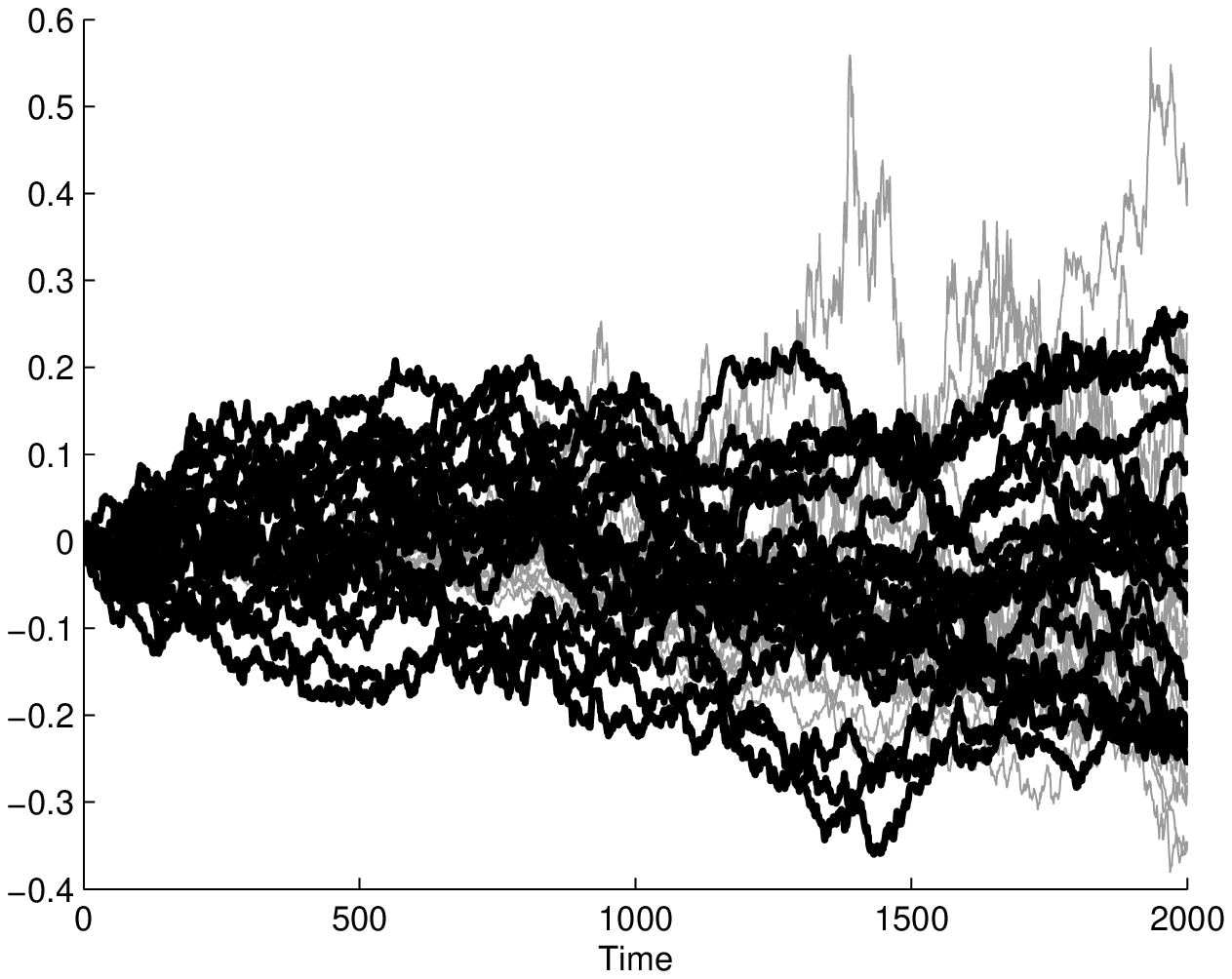}\includegraphics[scale=0.4]{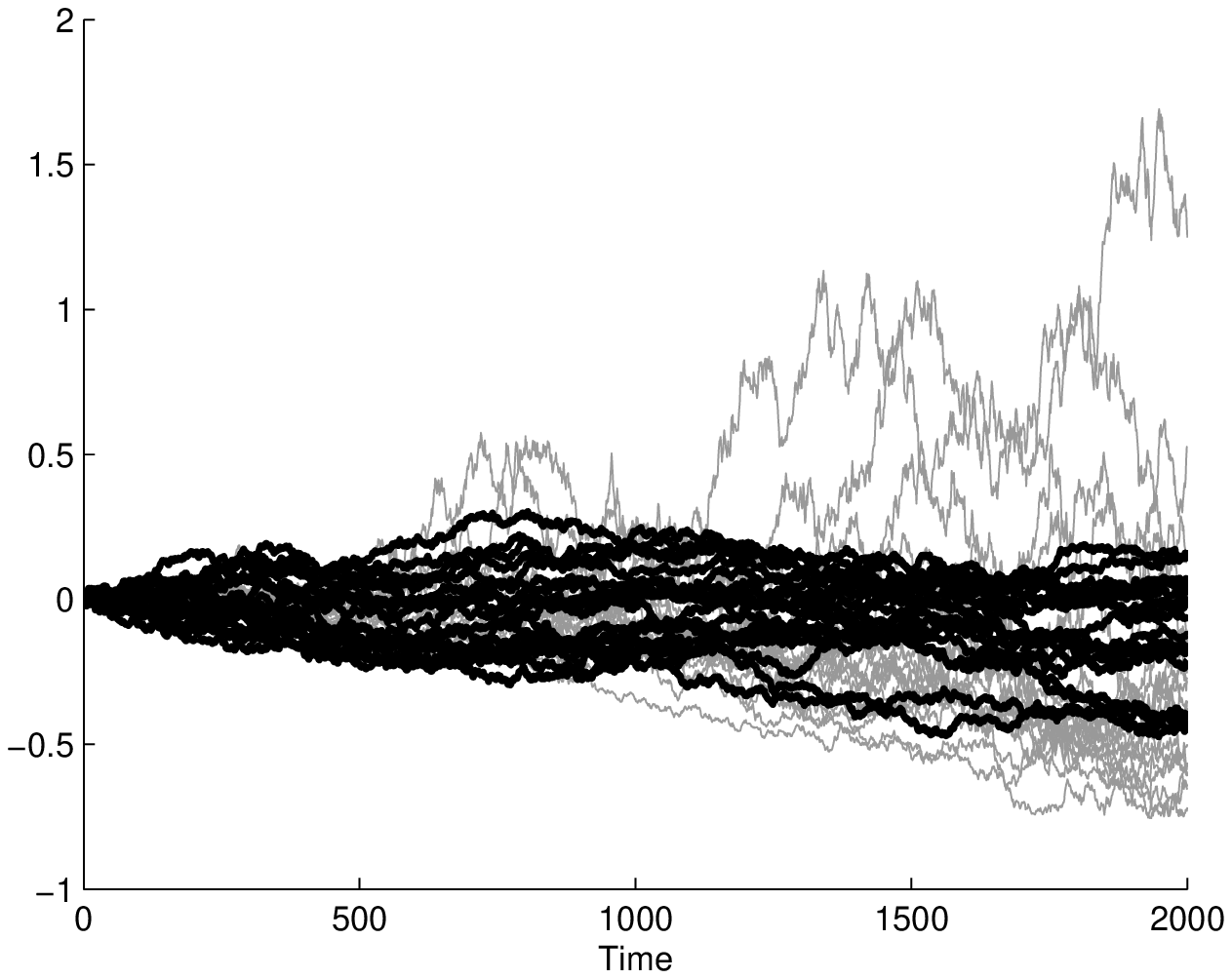}
\par\end{centering}

\protect\caption{\label{fig:ar1-example}$AR(1)$ example - The top two plots represent
the comparison of the estimates of $\mathbb{E}\left[\left(Z_{n}^{N}\right)^{2}\right]$
obtained using the standard MC approach (gray, thin lines) and the Pairs algorithm
(black, thick lines), where $N=50$ for the case of equal $M$ (top
left) and equal cost (top right) respectively. The bottom two represent
the comparison of the same two algorithms, but for the case, where
$N=250$ for the case of equal $M$ (bottom left) and equal cost (bottom
right)}
\end{figure*}

\begin{figure*}
\begin{centering}
\includegraphics[scale=0.4]{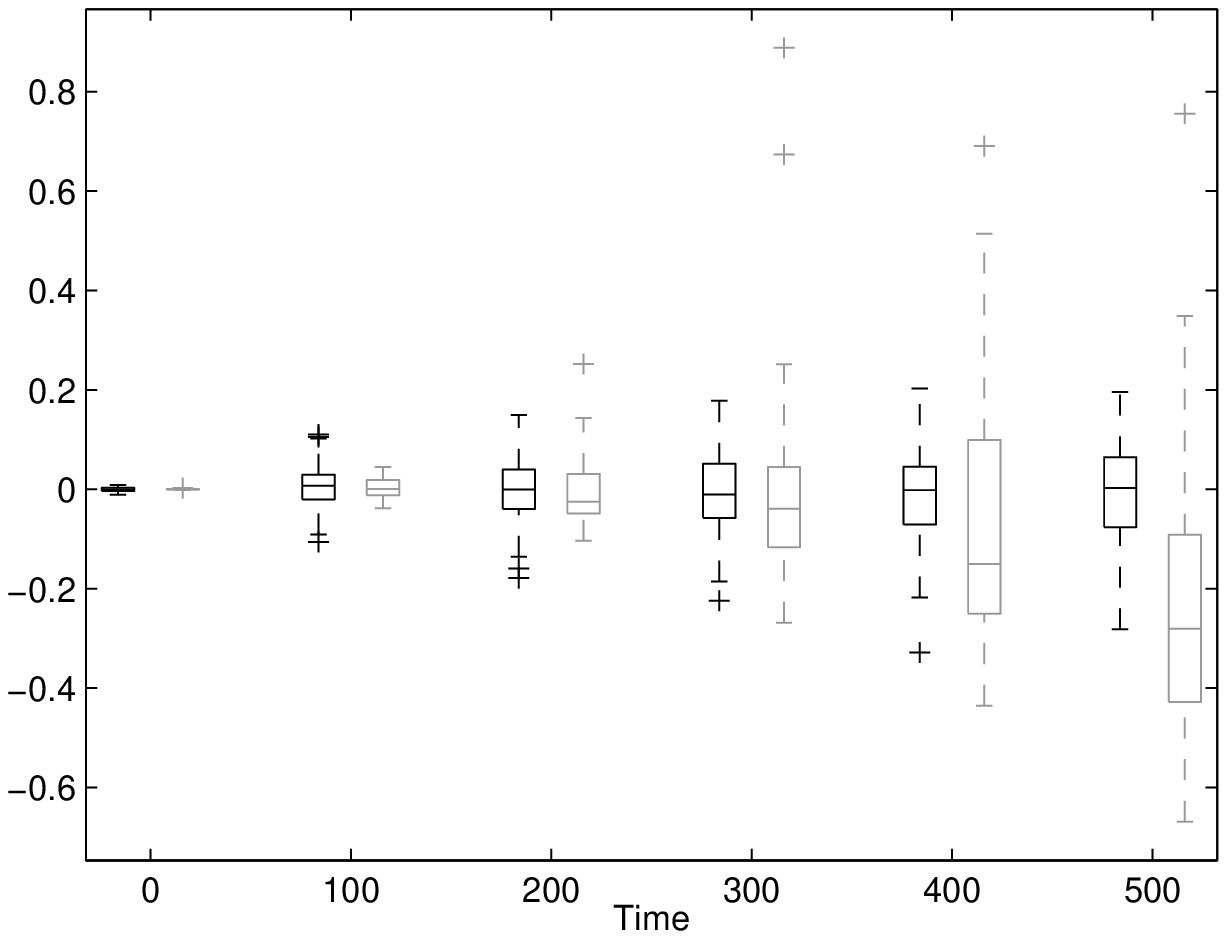}\includegraphics[scale=0.4]{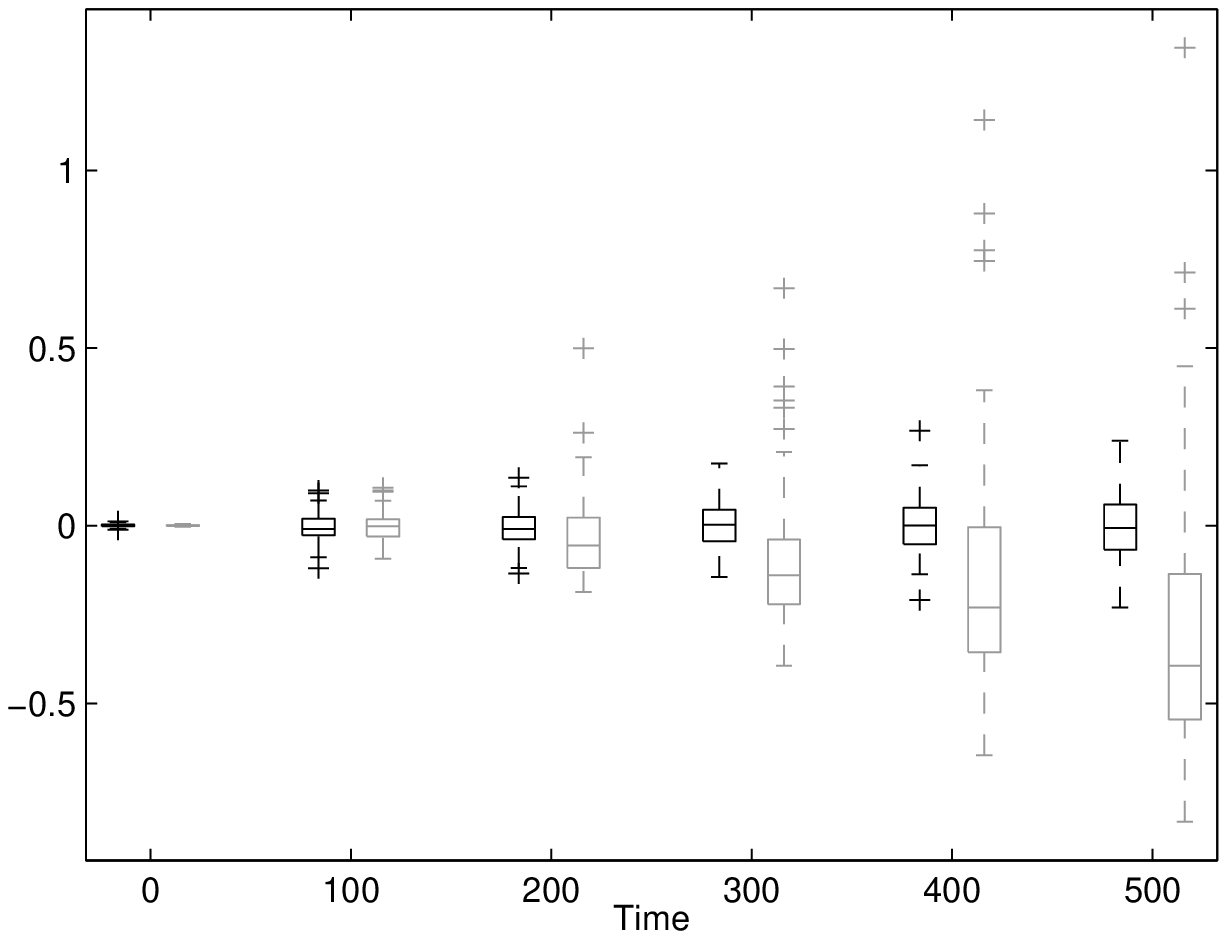}
\par
\end{centering}

\protect\caption{\label{fig:ar1-boxplots}$AR(1)$ example - boxplots using all of the 100 available simulated
paths for the case $N=50$ particles and equal $M$ and equal cost. The grey boxplots correspond to the MC approach, and black ones - to the
Pairs algorithm}
\end{figure*}

Figure \ref{fig:ar1-boxplots} shows boxplots based on 100 independent runs for both
algorithms for the case of equal $M=\tilde{M}=10^{4}$ and equal cost. We also have $N=50$. It is apparent that the estimates of $\mathbb{E}\left[\left(Z_{n}^{N}\right)^{2}\right]$
that we obtain using the Pairs algorithm have much less variability
than the estimates produced using the standard Monte Carlo approach with
i.i.d. replicates (especially for big values of the time parameter $n$).


%

\subsection{Lotka - Volterra system example}\label{sub:LV}

In this section we illustrate the numerical performance of the pairs
algorithm in the context of a partially observed Langevin approximation
to Lotka-Volterra ODE system \citep{golightly2011bayesian}. The signal
process in the HMM is obtained from a discretization of the stochastic
differential equation (SDE) $dX_{t}=\alpha(X_{t},c)dt+\sqrt{\beta(X_{t},c)}dW_{t}$,
where $X_{t}=(X_{1,t},X_{2,t})$, $W_{t}=(W_{1,t},W_{,2t})$. Here
$W_{t}$ is a vector, each of the components of which is independent
standard Brownian motion, $c=(c_{1},c_{2},c_{3})$ are parameters
and $\alpha(x,c)$ and $\beta(x,c)$ are the drift and diffusion coefficients
given for the Lotka-Volterra system by
\[
\alpha(x,c)=\left(\begin{array}{c}
c_{1}x_{1}-c_{2}x_{1}x_{2}\\
c_{2}x_{1}x_{2}-c_{3}x_{2}
\end{array}\right),\qquad\beta(x,c)=\left(\begin{array}{cc}
c_{1}x_{1}+c_{2}x_{1}x_{2} & -c_{2}x_{1}x_{2}\\
-c_{2}x_{1}x_{2} & c_{2}x_{1}x_{2}+c_{3}x_{2}
\end{array}\right),
\]
with $x=(x_{1},x_{2})$.

We consider Euler discretization of the SDE with time resolution $\Delta t=1/m$
for some $m\geq1$, with the resulting process satisfying
\begin{equation}
X_{n+(j+1)\Delta t}-X_{n+j\Delta t}=\alpha(X_{n+j\Delta t},c)\Delta t+\sqrt{\beta(X_{n+j\Delta t},c)\Delta t}\chi_{j}
\end{equation}
for $n\in\mathbb{N}$ and $j\in\{0,1,\ldots,m-1\}$, where $\chi_{j}$
is a sequence of $\mathcal{N}(0,1)$--independent random variables.
The signal process in the HMM, denoted by $(\mathbf{X}_{n})_{n\geq0}$,
consists of a $\mathbb{R}^{2}$--valued random variable $X_{0}=(100,100)$
and for $n\geq1$ a $\mathbb{R}^{2m}$--valued random variable $\mathbf{X}_{n+1}=(X_{n+\Delta t},X_{n+2\Delta t},\ldots,X_{n+1})$.
The model for the observations is $Y_{n}=X_{n}+\varepsilon_{n}$,
where $\varepsilon_{n}\sim\mathcal{N}(0,\Sigma_{2\times2})$ , $\Sigma_{2\times2}=\sigma^{2}I_{2\times2}$,
where $I_{2\times2}$ is the $2\times2$ identity matrix. We also assume that we have observed the
process at integer times $n$. Following
\citet{golightly2011bayesian}, we consider two values of the observation
noise variance $\sigma^{2}=10$ and $\sigma^{2}=200$. We fix the
rate constants $c=(c_{1},c_{2},c_{3})=(0.5,0.0025,0.3)$, and we will
use $m=1$ for the discretization parameter.

We adopt the same approach to constructing the proposal kernels $(q_{n})_{n\geq1}$
suggested in \citet[Section 4.3]{golightly2011bayesian}, in which
$q_{n}(\mathbf{x}_{n},\mathbf{x}_{n+1})$ is chosen to be a tractable
Gaussian approximation to the conditional density of $\mathbf{x}_{n+1}$
given $\mathbf{x}_{n}$,$\boldsymbol{y}_{n+1}$. The proposal kernel is given by
\[
q_{n+1}(\mathbf{x}_{n},\mathbf{x}_{n+1})=\prod_{j=0}^{m-1}\psi_{n+(j+1)\Delta t}(x_{n+j\Delta t},x_{n+(j+1)\Delta t})
\]
where $\psi_{n+(j+1)\Delta t}(x_{n+j\Delta t},\cdot)=\mathcal{N}(\cdot;x_{n+j\Delta t}+a_{j}\Delta t,b_{j}\Delta t)$,
where $a_{j}=\alpha_{j}+\beta_{j}(\beta_{j}\Delta_{j}+\Sigma)^{-1}(y_{n+1}-(x_{n+j\Delta t}+\alpha_{j}\Delta_{j}))$,
$b_{j}=\beta_{j}-\beta_{j}(\beta_{j}\Delta_{j}+\Sigma)^{-1}\beta_{j}\Delta t$,
$\Delta_{j}=1-j\Delta t$, $\alpha_{j}=\alpha(x_{n+j\Delta t},c)$,
$\beta_{j}=\beta(x_{n+j\Delta t},c)$. We consider the process $(\mathbf{X}_{n},Y_{n})_{n\geq0}$
as a HMM, to which the particle algorithms are applied to.
\begin{figure*}
\begin{centering}
\includegraphics[scale=0.4]{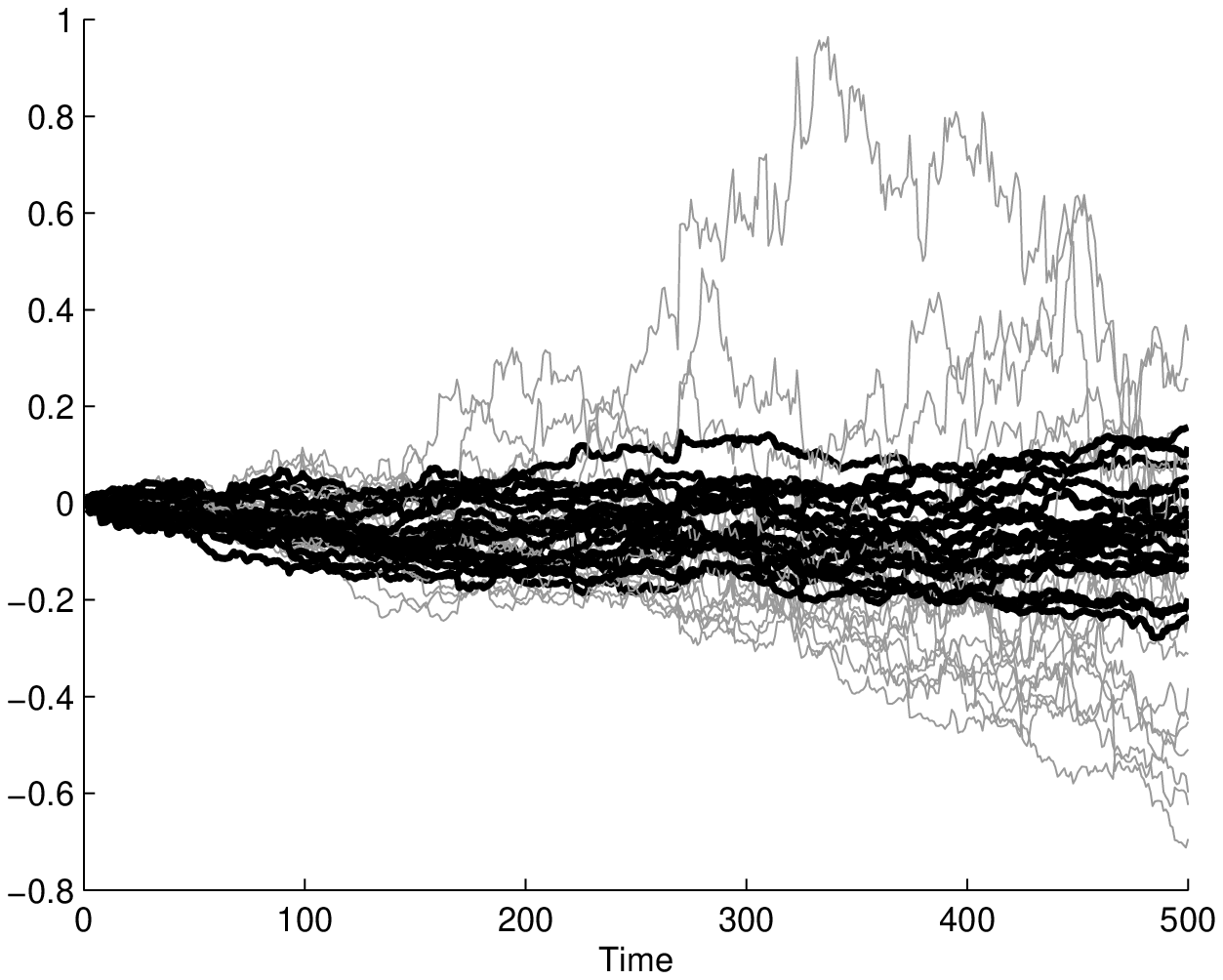}\includegraphics[scale=0.4]{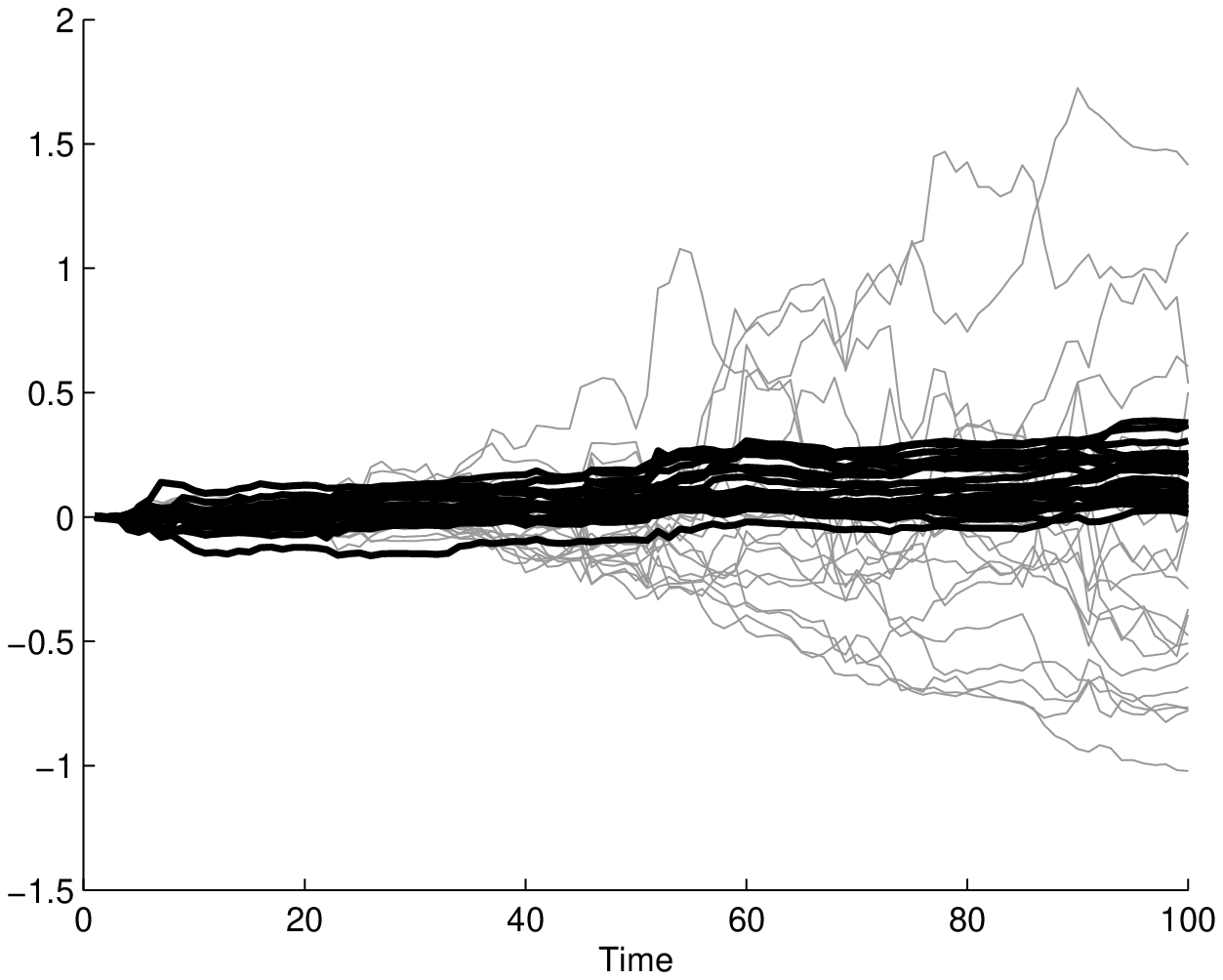}
\par
\end{centering}
\begin{centering}
\includegraphics[scale=0.4]{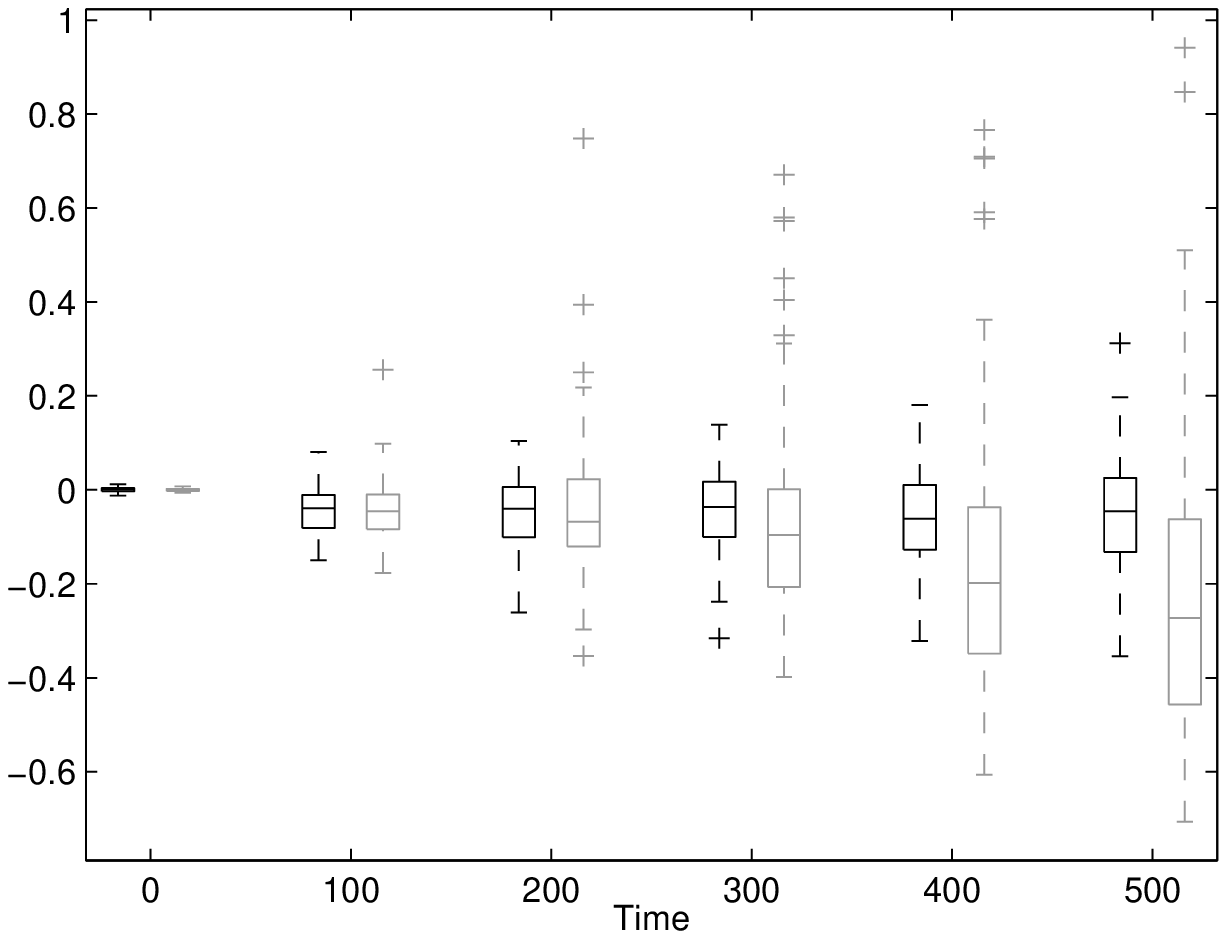}\includegraphics[scale=0.4]{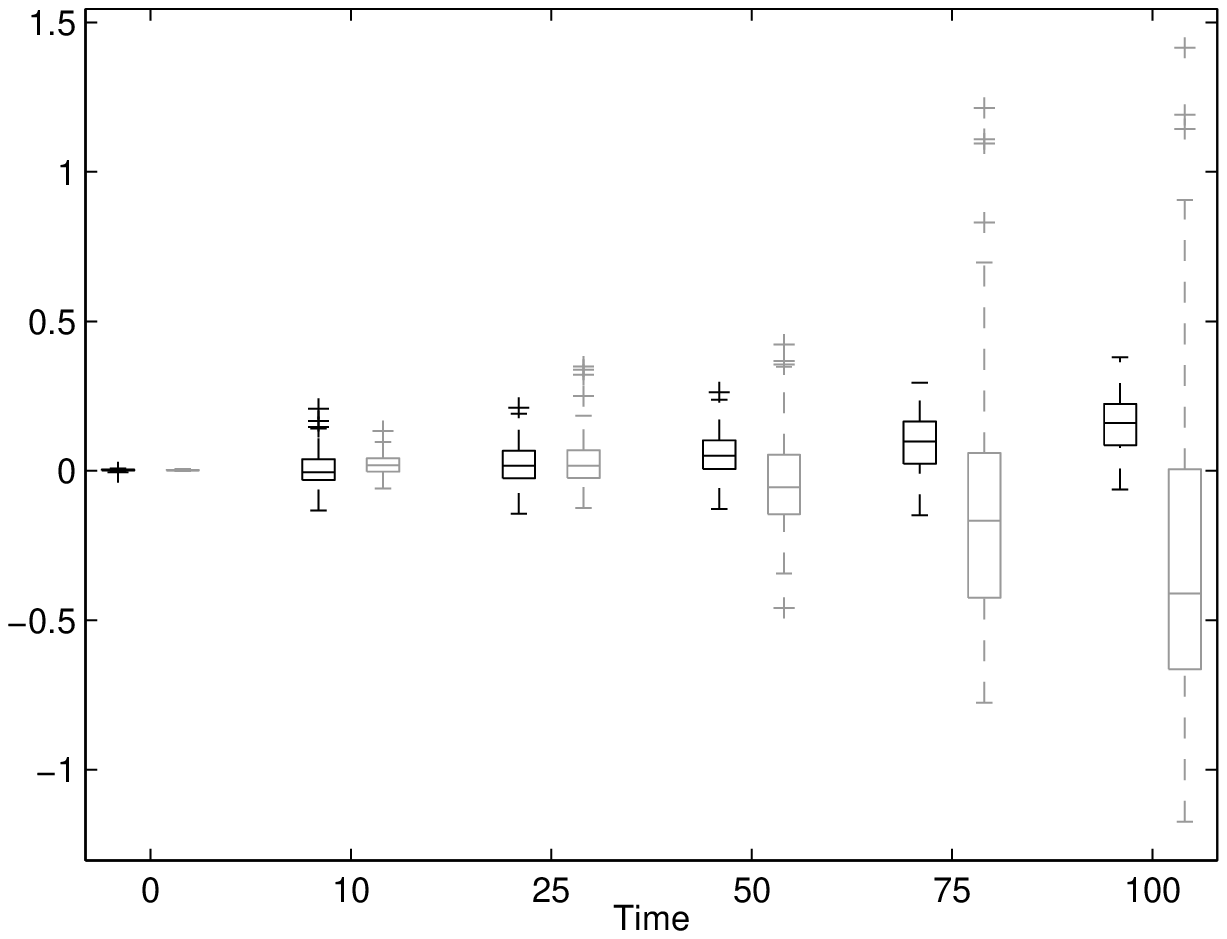}
\par
\end{centering}
\protect\caption{\label{fig:Lotka-Volterra-example}Lotka - Volterra example - comparison
of the estimates of $\mathbb{E}\left[\left(Z_{n}^{N}\right)^{2}\right]$ for the Pairs algorithm and
the standard Monte Carlo approach for the case of low observation noise ($\sigma^{2}=10$,
on the left) and large observation noise ($\sigma^{2}=200$, on the
right). The plots are for equal time cost. Again, grey corresponds to the MC approach and black
corresponds to the Pairs algorithm. The boxplots are based on 100 independent runs of the two
algorithms}
\end{figure*}

We first obtain a reliable benchmark value of $\mathbb{E}\left[\left(Z_{n}^{N}\right)^{2}\right]$,
denoted by $\Xi_{n}^{(N,M^{\prime})}$, using
a single run of the Pairs algorithm with $M^{\prime}=10^{6}$. We
compare $\Xi_{n}^{(N,M)}$ from the Pairs algorithm with
the simple Monte Carlo approximation $\widetilde{\Xi}_{n}^{(N,\tilde{M})}$
based on i.i.d. replicates, defined in (\ref{eq:Xi_tilde}) in Figure
\ref{fig:Lotka-Volterra-example} for two different values of the
observation noise - $\sigma^{2}=10$ and $\sigma^{2}=200$. In both
cases we plot again $\log(\Xi_{n}^{\left(N,M\right)})-\log(\Xi_{n}^{(N,M^{\prime})})$
for the Pairs algorithm and $\log(\widetilde{\Xi}_{n}^{(N,\tilde{M})})-\log(\Xi_{n}^{(N,M^{\prime})})$
for the standard MC approach.

On the top left of Figure \ref{fig:Lotka-Volterra-example} we have the
low noise example. In this example, we set $N=100$, $M=10^{4}$ and
$\tilde{M}=300$. On the top right plot we present the large noise case where
we set $N=100$, $M=10^{5}$ and $\tilde{M}=3000$ in order to equalize the computational cost.
Again, as in the previous example, we have plotted 20 independent runs for both algorithms.

In the two plots, and especially for large values of the time parameter $n$, the
estimate that we obtain with the help of the Pairs algorithm has much
less variability than the estimate calculated using standard Monte
Carlo with i.i.d. replicates. We can clearly see that with the increase
of the time parameter $n$, the rate of growth of the variability of
the estimates of $\mathbb{E}\left[\left(Z_{n}^{N}\right)^{2}\right]$
obtained using the Pairs algorithm is far less than the corresponding rate
for the standard Monte Carlo approach (using i.i.d. replicates). The observations
about the variability of the estimates in Figure \ref{fig:Lotka-Volterra-example}
are also supported by the corresponding boxplots, based on 100 independent runs
of the two algorithms.

\begin{table}
\begin{centering}
\begin{tabular}{ccccc}
\hline
\multirow{1}{*}{Time} & \multicolumn{2}{c}{Low noise} & \multicolumn{2}{c}{Large noise}\tabularnewline
\cline{2-5}
$n$ & $Z_{n}^{N}$ & $\Xi_{n}^{\left(N,M\right)}$ & $Z_{n}^{N}$ & $\Xi_{n}^{\left(N,M\right)}$\tabularnewline
\hline
$1$ & $7.4\times10^{-3}$ & $4.9\times10^{-5}$ & $2.21\times10^{-4}$ & $6.3\times10^{-8}$\tabularnewline
$5$ & $4.13\times10^{-15}$ & $1.6\times10^{-29}$ & $4.07\times10^{-21}$ & $3.9\times10^{-41}$\tabularnewline
$10$ & $2.51\times10^{-33}$ & $4.9\times10^{-66}$ & $4.59\times10^{-42}$ & $1.15\times10^{-82}$\tabularnewline
$25$ & $2.02\times10^{-80}$ & $4.9\times10^{-160}$ & $3.42\times10^{-100}$ & $2.69\times10^{-198}$\tabularnewline
$50$ & $2.23\times10^{-159}$ & $8.9\times10^{-318}$ & $2.81\times10^{-195}$ & $\leq10^{-324}$\tabularnewline
$100$ & $6.41\times10^{-317}$ & $\leq10^{-324}$ & $\leq10^{-324}$ & $\leq10^{-324}$\tabularnewline
\hline
\end{tabular}
\par\end{centering}

\raggedright{}\protect\caption{\label{tab:Normalizing-constants-and-errors}Estimates of $Z_{n}$
and $\mathbb{E}\left[\left(Z_{n}^{N}\right)^{2}\right]$ (using the
Pairs algorithm, hence $\Xi_{n}^{\left(N,M\right)}$ with $M=10^{6}$)
for the two cases of low and large observation noise for the Lotka-Volterra example}
\end{table}

Table \ref{tab:Normalizing-constants-and-errors} shows numerical
values for $Z_{n}^{N}$ and $\Xi_{n}^{\left(N,M\right)}$ for different
values of the time parameter $n$ for the Lotka--Volterra example. We see, that although the scale
of the values in Table \ref{tab:Normalizing-constants-and-errors}
is small, we still have, by Jensen's inequality, that $\mathbb{E}\left[\left(Z_{n}^{N}\right)^{2}\right]\geq\mathbb{E}\left[Z_{n}^{N}\right]^{2}$.

\subsection{Estimating Monte Carlo variance}
\label{sub:Strategies_Comparison}

The purpose of this example is to show that the benefits of approximating $\mathbb{E}\left[\left(Z_{n}^{N}\right)^{2}\right]$ using the Pairs algorithm carry over to its use within a strategy for both estimating $Z_n$ and reporting Monte Carlo variance. As a benchmark for comparisons, we consider the following standard approach based on i.i.d. replicates of a particle filter.

\paragraph{\underline{MC strategy.}}
Run $\tilde{M}$ independent particle filters, each with $\tilde{N}$ particles, to give $\left\{Z_{n}^{\tilde{N},j}\right\}_{j=1}^{\tilde{M}}$. Then report:
\begin{itemize}
  \item $\widetilde{Z}_{n}^{(\tilde{N},\tilde{M})}=
  \frac{1}{\tilde{M}}\sum_{j=1}^{\tilde{M}}Z_{n}^{\tilde{N},j}$ as an estimate of $Z_{n}$
  \item $\frac{1}{\tilde{M}}
  \frac{1}{\tilde{M}-1}\sum_{j=1}^{\tilde{M}}\left(Z_{n}^{\tilde{N},j}-\widetilde{Z}_{n}^{
  (\tilde{N},\tilde{M})}\right)^{2}$ as an estimate of $\mathrm{Var}\left[\widetilde{Z}_{n}^{(\tilde{N},\tilde{M})}
  \right]$
\end{itemize}
The cost of this strategy is $O(\tilde{N} \tilde{M})$, and the variance estimate it delivers is a standard sample variance, thus unbiased. There are various ways that the MC strategy could be changed or augmented by using the Pairs algorithm. We consider the following:

\paragraph{\underline{Pairs strategy.}}
Run $M$ independent particle filter algorithms, each with $N$ particles, to give $\left\{Z_{n}^{N,j}\right\}_{j=1}^{N}$. Additionally run one instance of the Pairs algorithm with parameters $(M,N)$, to give $\Xi_{n}^{\left(N,M\right)} $. Then report:
\begin{itemize}
  \item  ${Z}_{n}^{(N,M)}=
  \frac{1}{M}\sum_{j=1}^{M}Z_{n}^{N,j}$ as an estimate of $Z_{n}$
  \item  $\frac{1}{M-1}
  \left[\Xi_{n}^{\left(N,M\right)} - \left({Z}_{n}^{(N,M)}\right)^{2}\right]$ as an estimate of $\mathrm{Var}\left[{Z}_{n}^{(N,M)} \right]$
\end{itemize}
The cost of this strategy is $O(MN+M)$. So if for instance $N=\tilde{N}$ and $M=\tilde{M}$, the additional cost of the Pairs strategy beyond that of the MC strategy becomes negligible as $N$ grows.

To see that the variance estimate delivered by the Pairs strategy is unbiased, note that:
\begin{align*}
&\frac{M}{M-1}\mathbb{E}\left[\Xi_{n}^{(M,N)}-\left(Z_{n}^{(N,M)}\right)^{2}\right] \\& =\frac{M}{M-1}\left[\mathbb{E}\left[\Xi_{n}^{(M,N)}\right]-\frac{1}{M^{2}}\sum_{j=1}^{M}\mathbb{E}\left[\left(Z_{n}^{N,j}\right)^{2}\right]-\frac{1}{M^{2}}\sum_{i\neq j}^{M}\mathbb{E}\left[Z_{n}^{N,i}\right]\mathbb{E}\left[Z_{n}^{N,j}\right]\right]\\
 & =\frac{M}{M-1}\left[\mathbb{E}\left[\left(Z_{n}^{N}\right)^{2}\right]-\frac{1}{M}\mathbb{E}\left[\left(Z_{n}^{N}\right)^{2}\right]-\left(1-\frac{1}{M}\right)\mathbb{E}\left[Z_{n}^{N}\right]^{2}\right]\\
 & =\mathrm{Var}[Z_{n}^{N}]=M\mathrm{Var}\left[\frac{1}{M}\sum_{j=1}^{M}Z_{n}^{(N,j)}\right],
\end{align*}
where the second equality uses the lack-of-bias property of the Pairs algorithm from Theorem \ref{thm:Pairs-convergence-and-bound}, i.e. $\mathbb{E}\left[\Xi_{n}^{(M,N)}\right]=\mathbb{E}\left[\left(Z_{n}^{N}\right)^{2}\right]$.

Numerical results are shown in Figure \ref{fig:ar1-strategies-comparison}. In order to achieve better visual representation, we plot normalized estimates ${Z}_{n}^{(N,M)}/Z_{n}^{N^{\prime}}$ and $\widetilde{Z}_{n}^{(\tilde{N},\tilde{M})}/Z_{n}^{N^{\prime}}$ and their variances $\mathrm{Var}\left[{Z}_{n}^{(N,M)}\right]/\left(Z_{n}^{N^{\prime}}\right)^{2}$ and $\mathrm{Var}\left[\widetilde{Z}_{n}^{(\tilde{N},\tilde{M})}\right]/\left(Z_{n}^{N^{\prime}}\right)^{2}$, where $Z_{n}^{N^{\prime}}$ is a reliable, benchmark estimate of $Z_n$ obtained from a particle filter with $N^{\prime}=10^{6}$. We make comparisons with $N=\tilde{N}=50$ and $M=\tilde{M}=10^4$, with these settings in our implementation the additional cost of the Pairs strategy beyond that of the MC strategy was found to be insignificant, very similar results were obtained if the costs of the two strategies were exactly equalized.

In Figure \ref{fig:ar1-strategies-comparison} we compare the MC and Pairs strategies. The top left shows box plots of ${Z}_{n}^{(N,M)}/Z_{n}^{N^{\prime}}$ and $\widetilde{Z}_{n}^{(\tilde{N},\tilde{M})}/Z_{n}^{N^{\prime}}$ obtained from $1000$ independent realizations of the two strategies, for different values of $n$. The top right shows boxplots for the variance estimates, also from $1000$ realizations. We can clearly see that for increasing $n$ the estimates for the MC strategy exhibit larger variability than the estimates obtained from the Pairs strategy. On the bottom two plots of Figure \ref{fig:ar1-strategies-comparison} we compare the kernel density estimates for
of $\mathrm{Var}\left[{Z}_{n}^{(N,M)}\right]/\left(Z_{n}^{N^{\prime}}\right)^{2}$ and $\mathrm{Var}\left[\widetilde{Z}_{n}^{(\tilde{N},\tilde{M})}\right]/\left(Z_{n}^{N^{\prime}}\right)^{2}$ for $n=500$. On bottom left the estimated density is plotted, and on bottom right the $\log$ of the density is plotted, highlighting the heavier tails of the distribution for the MC strategy. The kernel density estimates in both plots were produced using a normal kernel function with bandwidths 0.06 (Pairs strategy) and 0.9 (MC strategy). The density estimates indicated a more concentrated distribution for the Pairs strategy (thick, black line) than for the MC strategy (grey line).

\begin{figure*}
\begin{centering}
\includegraphics[scale=0.4]{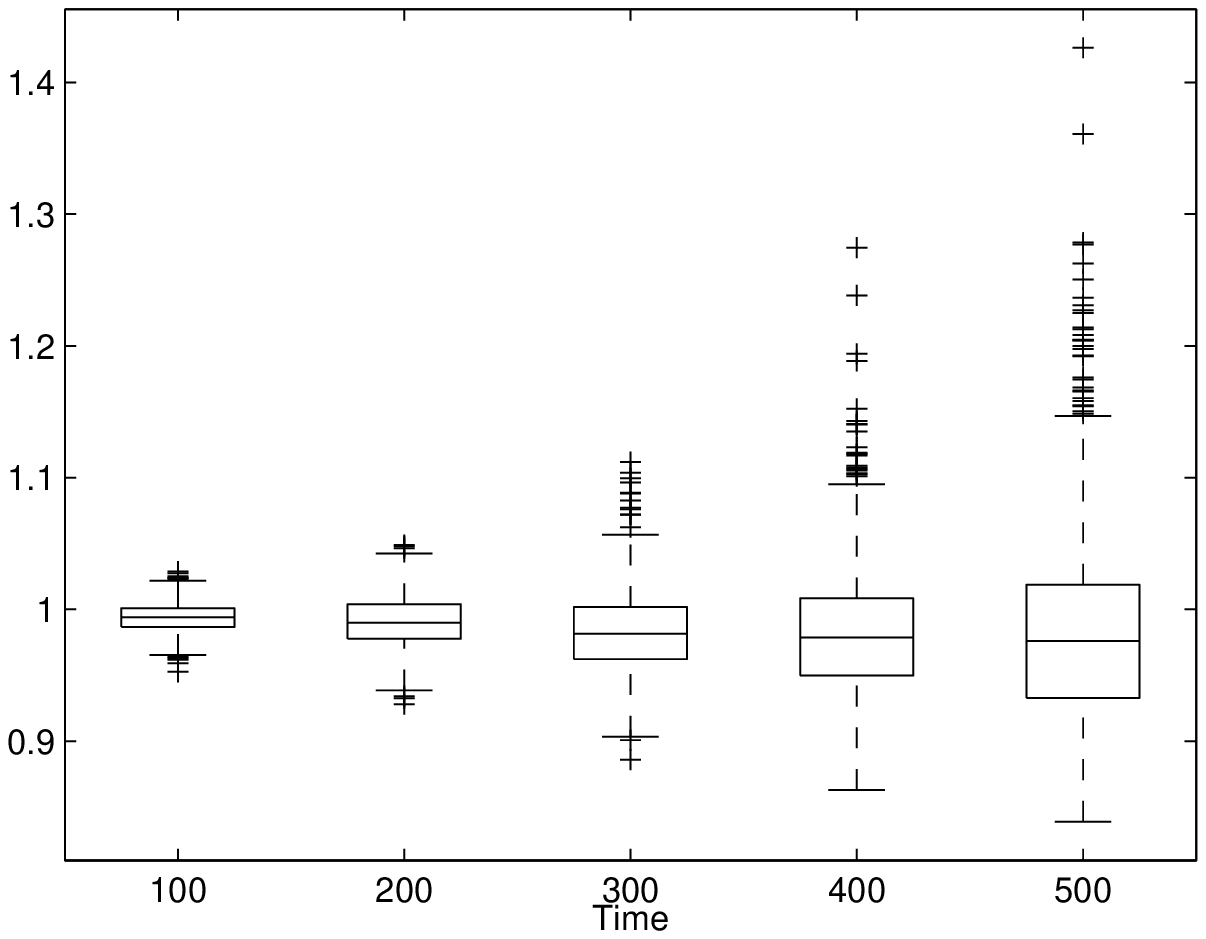}\includegraphics[scale=0.4]{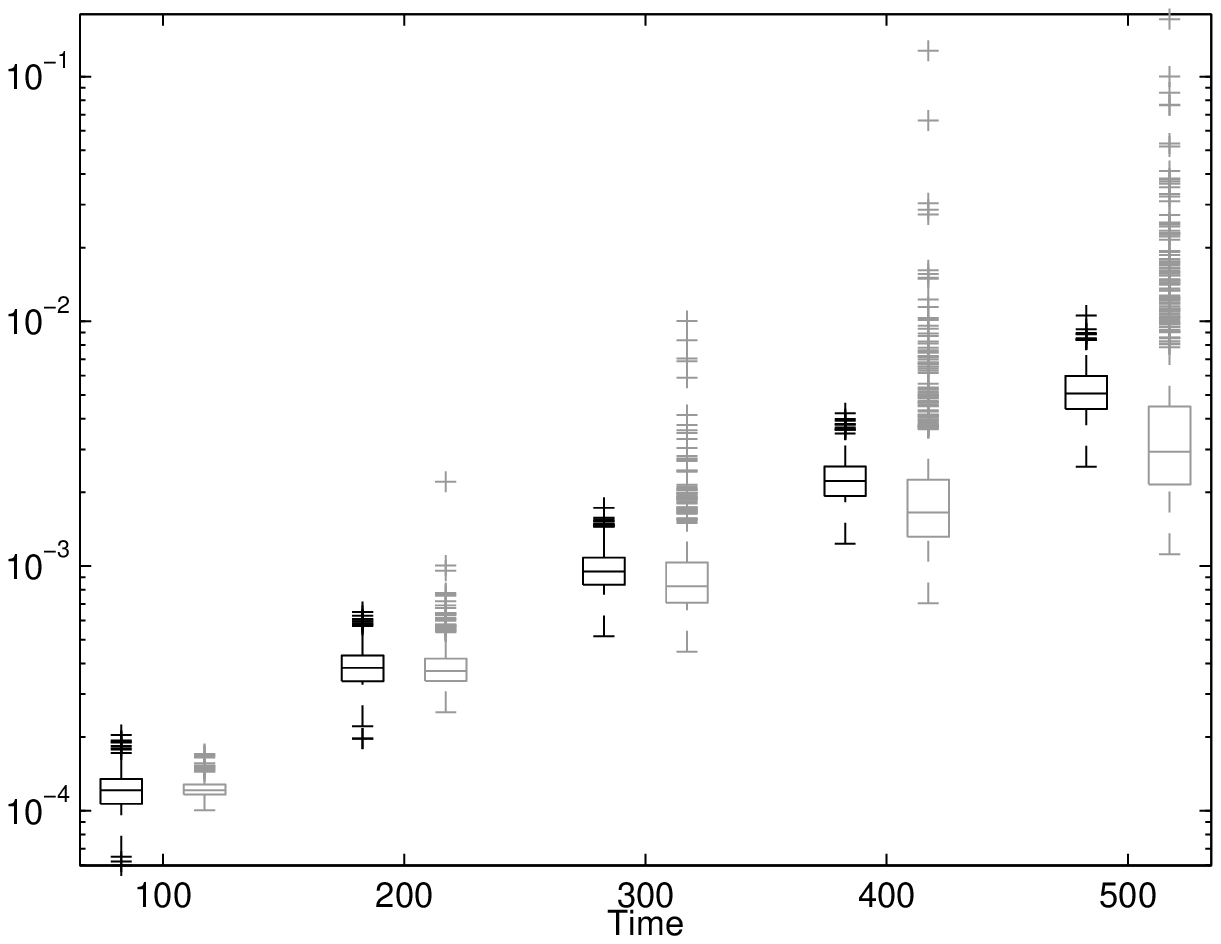}
\par\end{centering}

\begin{centering}
\includegraphics[scale=0.4]{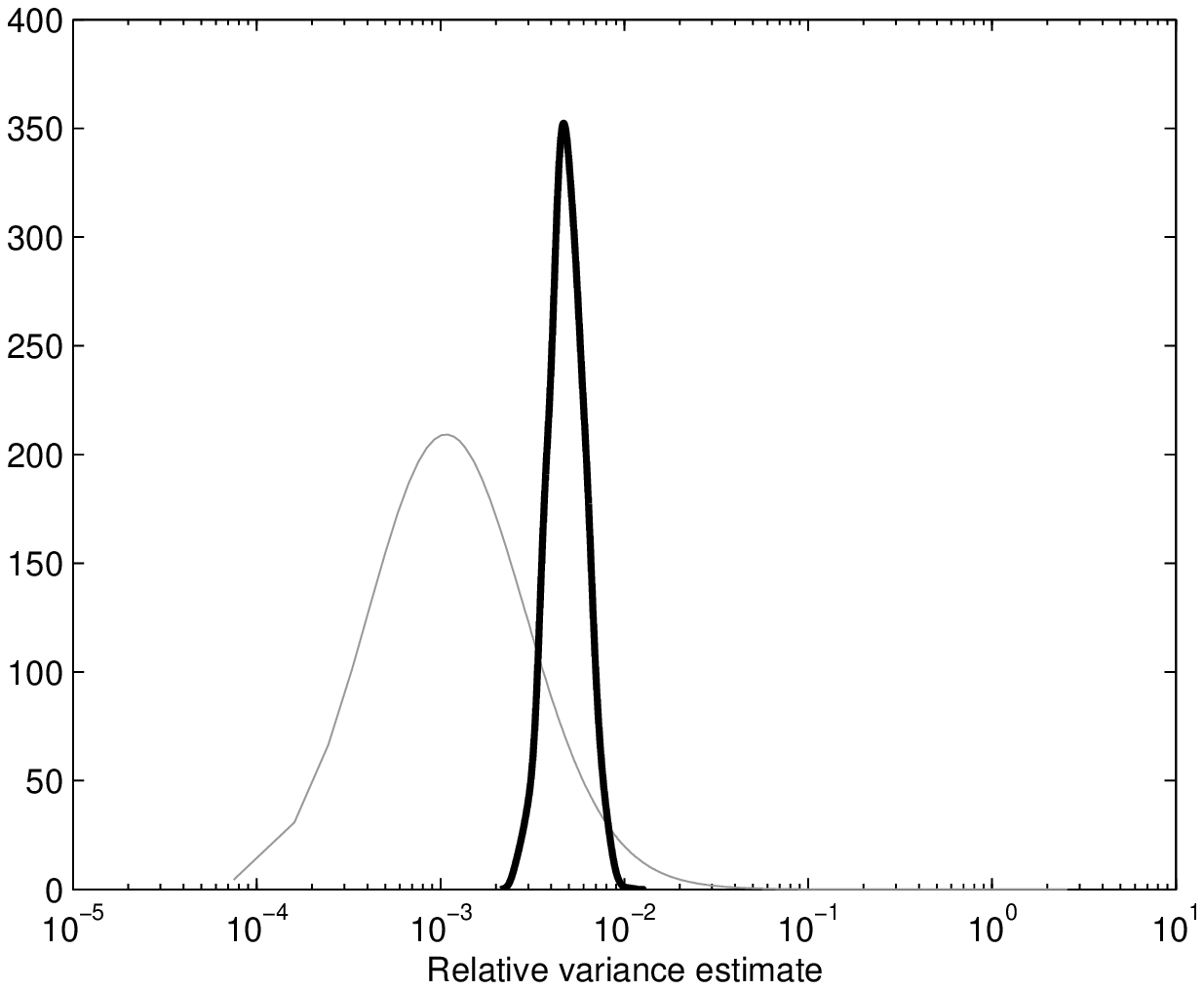}\includegraphics[scale=0.4]{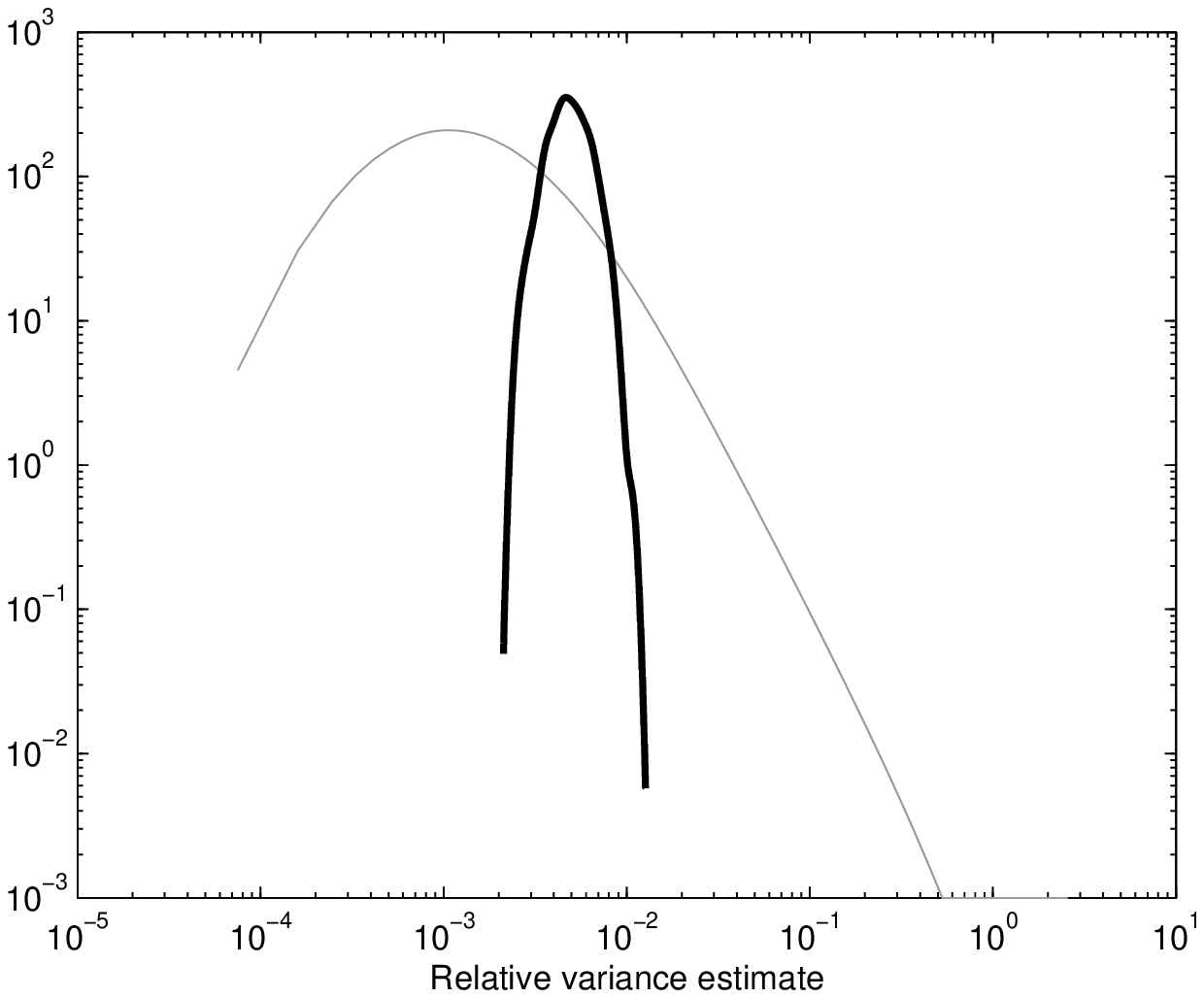}
\par\end{centering}

\protect\caption{\label{fig:ar1-strategies-comparison}$AR(1)$ example - comparison of Pairs and MC strategies for $\tilde{N}=N=50$ particles and $\tilde{M}=M=10^4$. On top left we plot the estimates of $\mathbb{E}\left[Z_{n}^{N}\right]/Z_{n}^{N^{\prime}}$ for both MC and Pairs strategies (which are equal). On top right plot we compare the two strategies in terms of estimates of the relative variance $Var\left[{Z}_{n}^{(N,M)} \right]/\left(Z_{n}^{N^{\prime}}\right)^{2}$ and $Var\left[\widetilde{Z}_{n}^{(\tilde{N},\tilde{M})}\right]/\left(Z_{n}^{N^{\prime}}\right)^{2}$ respectively (the $y$--axis is on a $\log$--scale). On the bottom left (right) plot we compare the kernel density estimates of the pdf ($\log$--pdf) of the relative variance for the two strategies for time $n=500$ (the $y$--axis of the bottom right plot is on a $\log$--scale). For the bottom two plots the $x$--axis is on a $\log$--scale}
\end{figure*}

\appendix

\section{Auxiliary definitions, results and proof of
Theorem \ref{thm:Pairs-convergence-and-bound}}
\label{sec:Appendix}

This appendix is structured as follows. After introducing notation
in \ref{sub:Notation-and-conventions}, \ref{sub:A-generic-particle}
introduces a generic particle system, of which we show Algorithm \ref{alg:standard-SMC-algorithm}
to be a special case. The account of this generic particle system
and some of its properties is needed in order to derive an associated
pairs particle system in \ref{sub:The-Pairs-particle-system}, of
which we show Algorithm \ref{alg:PairsAlgo} to be a special case.
The proof of Theorem \ref{thm:Pairs-convergence-and-bound}, in \ref{sub:Proof-of-Theorem},
rests on the key observation that the pairs particle system is also
an instance of the generic particle system of \ref{sub:A-generic-particle},
allowing properties of the latter to be transferred to the Pairs algorithm.

\subsection{Notation and conventions}
\label{sub:Notation-and-conventions}

For a measurable space $(E,\mathcal{E})$, denote by $\mathcal{B}_{b}(E)$
the set of all $\mathbb{R}$-valued, measurable and bounded functions
on $E$, and by $\mathcal{M}(E)$ and $\mathcal{P}(E)$ the sets of
respectively measures and probability measures on $\mathcal{E}$.
For $\mu\in\mathcal{M}(E)$ and $\varphi\in\mathcal{B}_{b}(E)$ we
write $\mu(\varphi):=\int_{E}\varphi(x)\mu(dx)$. For a non-negative
integral kernel $L:E\times\mathcal{E}\rightarrow[0,\infty)$, $\varphi\in\mathcal{B}_{b}(E)$
and $\mu\in\mathcal{M}(E)$, we write $L(\varphi)(x):=\int_{E}L(x,dy)\varphi(y)$,
$\left(\mu L\right)(\cdot):=\int_{E}\mu(dx)L(x,\cdot)$ and for two
such kernels, $L$ and $M$, we write their composition as $(LM)(x,\cdot):=\int_{E}L(x,dx^{\prime})M(x,\cdot)$.
We write two-fold tensor product measures and functions as respectively
$\mu^{\otimes2}\in\mathcal{M}(E^{2})$ and $\varphi\otimes\varphi\in\mathcal{B}_{b}(E\times E)$.
For $\varphi\in\mathcal{B}_{b}(E\times E)$ we write the tensor product
integral operator $L^{\otimes2}(\varphi)(x,x^{\prime}):=\int_{E\times E}L(x,dy)L(x^{\prime},dy^{\prime})\varphi(y,y^{\prime})$.
We introduce also a measurable space $(E_{0},\mathcal{E}_{0})$ and
use exactly similar notation when dealing with functions, measures
and kernels on $(E_{0},\mathcal{E}_{0})$, and kernels between $(E_{0},\mathcal{E}_{0})$
and $(E,\mathcal{E})$.

\subsection{A generic particle system}
\label{sub:A-generic-particle}

For each $n\geq2$ let $Q_{n}:E\times\mathcal{E}\rightarrow(0,\infty)$
be an integral kernel such that for each $x\in E$, $Q_{n}(x,\cdot)$
is a finite measure on $(E,\mathcal{E})$. Then introduce
\begin{equation}
M_{n}:\;(x,A)\in E\times\mathcal{E}\;\mapsto\;\frac{Q_{n}(x,A)}{Q_{n}(x,E)}\in[0,1];\quad G_{n-1}:\; x\in E\;\mapsto\; Q_{n}(x,E)\in(0,\infty),\label{eq:M_and_G_from_Q}
\end{equation}
which are respectively a Markov kernel and a measurable, bounded,
strictly positive function. Let also $Q_{1}:E_{0}\times\mathcal{E}\rightarrow(0,\infty)$
be a finite integral kernel, with $M_{1}$ and $G_{0}$ defined similarly
to (\ref{eq:M_and_G_from_Q}).

For $0\leq p\leq n$ define $Q_{p,n}=Q_{p+1}\cdots Q_{n}$ with $Q_{n,n}:=Id$.
Fix some $\eta_{0}\in\mathcal{P}(E_{0})$, and define the measures
$\left(\gamma_{n}\right)_{n\geq0}$ and probability measures $\left(\eta_{n}\right)_{n\geq1}$
by $\gamma_{0}:=\eta_{0}$ and
\begin{equation}
\gamma_{n}(\cdot):=\eta_{0}Q_{0,n}(\cdot),\quad\eta_{n}(\cdot):=\frac{\gamma_{n}(\cdot)}{\gamma_{n}(E)},\quad n\geq1.\label{eq:gamma_n_defn}
\end{equation}
With these objects, and for some fixed $N\geq1$, we associate a particle
process $\left(\zeta_{n}\right)_{n\geq0}$ as follows. The initial
configuration $\zeta_{0}=\left\{ \zeta_{0}^{i}\right\} _{i=1}^{N}$
are independent and identically distributed according to $\eta_{0}$,
and the evolution of $\zeta_{n}=\left\{ \zeta_{n}^{i}\right\} _{i=1}^{N}$
is described by the following probability law
\begin{eqnarray}
\mathbb{P}(\left.\zeta_{n}\in d\zeta_{n}\right|\zeta_{0},...,\zeta_{n-1}): & = & \prod_{i=1}^{N}\frac{\sum_{j=1}^{N}Q_{n}(\zeta_{n-1}^{j},d\zeta_{n}^{i})}{\sum_{j=1}^{N}Q_{n}(\zeta_{n-1}^{j},E)}\label{eq:generic_law}\\
 & = & \prod_{i=1}^{N}\frac{\sum_{j=1}^{N}G_{n-1}(\zeta_{n-1}^{j})M_{n}(\zeta_{n-1}^{j},d\zeta_{n}^{i})}{\sum_{j=1}^{N}G_{n-1}(\zeta_{n-1}^{j})},\quad n\geq1,\nonumber
\end{eqnarray}
where $d\zeta_{n}$ is to be understood as an infinitesimal neighborhood
of a point $(\zeta_{n}^{1},...,\zeta_{n}^{N})$.

Let us define the empirical measures
\begin{eqnarray}
 &  & \eta_{n}^{N}:=N^{-1}\sum_{i=1}^{N}\delta_{\zeta_{n}^{i}},\; n\geq0.\label{eq:particle-measures}\\
 &  & \gamma_{0}^{N}:=\eta_{0}^{N},\quad\gamma_{n}^{N}(\cdot):=\eta_{n}^{N}(\cdot)\prod_{p=0}^{n-1}\eta_{p}^{N}(G_{p}),\quad n\geq1.\label{eq:gamma_N_defn}
\end{eqnarray}

\subsubsection*{Algorithm \ref{alg:standard-SMC-algorithm}
as an instance of the generic particle system}
\label{sub:Algo1-as-GPS}

Let $\left(\mathsf{X},\mathcal{X}\right)$, $\pi_{0},$ $f$, $g$
be the ingredients of the HMM as in Section \ref{sec:SMC-and-HMM}.
To obtain Algorithm \ref{alg:standard-SMC-algorithm} as an instance
of the generic particle system under the law (\ref{eq:generic_law}),
take $E_{0}=\mathsf{X}$, $\mathcal{E}_{0}=\mathcal{X}$, and $E=\mathsf{X}^{2}$
, $\mathcal{E}=\mathcal{X}^{\otimes2}$. Then for points $x=(x_{1},x_{2})\in E$
and $y=(y_{1},y_{2})\in E$, take
\begin{equation}
M_{n}(x,dy)=\delta_{x_{2}}(dy_{1})q_{n}(y_{1},y_{2})dy_{2},\quad G_{n-1}(x)=\frac{g_{n-1}(x_{2})f(x_{1},x_{2})}{q_{n-1}(x_{1},x_{2})},\quad n\geq2,\label{eq:Algo1-as-GPS_1}
\end{equation}
and for $x\in E_{0}$, $y=(y_{1},y_{2})\in E$, take
\begin{equation}
M_{1}(x,dy)=\delta_{x}(dy_{1})q_{1}(y_{1},y_{2})dy_{2},\quad G_{0}(x)=\frac{g_{0}(x)\pi_{0}(x)}{q_{0}(x)},\quad\eta_{0}=\pi_{0}.\label{eq:Algo1-as-GPS_2}
\end{equation}
 Observe then that with $Z_{n}$ as in (\ref{eq:Z_n_defn}) and $Z_{n}^{N}$
as in Algorithm \ref{alg:standard-SMC-algorithm},
\begin{equation}
\gamma_{n+1}(1)\equiv Z_{n},\quad\gamma_{n+1}^{N}(1)\equiv Z_{n}^{N}.\label{eq:gamma_equiv_Z}
\end{equation}

\subsubsection*{Properties of the generic particle system}

We now give a brief account of certain key properties of the particle
system introduced above, which we shall later put to use in analyzing
the pairs algorithm.
\begin{rem}
It is known that when, for each $n\geq0$,
\begin{equation}
\sup_{x}G_{n}(x)<\infty,\label{eq:G_bounded}
\end{equation}
 we have for any $\varphi\in\mathcal{B}_{b}(E)$,
\begin{equation}
\eta_{n}^{N}(\varphi)\stackrel[N\rightarrow\infty]{a.s.}{\longrightarrow}\eta_{n}(\varphi),\quad\quad\gamma_{n}^{N}(\varphi)\stackrel[N\rightarrow\infty]{a.s.}{\longrightarrow}\gamma_{n}(\varphi),\label{eq:standard-Convergence}
\end{equation}
see e.g. \citep[Theorem 7.4.2]{DelMoral04}. Moreover, as discussed
in \citep[ Section 9.4.1]{DelMoral04},
\begin{equation}
\mathbb{E}\left[\gamma_{n}^{N}(\varphi)\right]=\gamma_{n}(\varphi)=\eta_{0}Q_{0,n}(\varphi),\quad\forall N\geq1.\label{eq:particle-filter-unbiasedness}
\end{equation}

\end{rem}
\citet{CerouDelMoral2011} have obtained second moment formulae for
$\gamma_{n}^{N}(1)$ via a study of the tensor product empirical measures:
\begin{eqnarray*}
 &  & \left(\eta_{n}^{N}\right)^{\otimes2}:=\frac{1}{N^{^{2}}}\sum_{i=1}^{N}\sum_{j=1}^{N}\delta_{\zeta_{n}^{i}}\otimes\delta_{\zeta_{n}^{j}}\\
 &  & \left(\gamma_{n}^{N}\right)^{\otimes2}:=\gamma_{n}^{N}(1)^{2}\left(\eta_{n}^{N}\right)^{\otimes2}.
\end{eqnarray*}

Introducing the coalescence operator $C$ which acts on bounded measurable
functions $F$ as $C(F)(x,y)=F(x,x)$, we have:
\begin{prop}
\label{prop:Cerou3.4-1} \citep[Lemma 3.2]{CerouDelMoral2011}For
any $F\in\mathcal{B}_{b}(E\times E)$,
\begin{equation}
\mathbb{E}\left[\left(\gamma_{n}^{N}\right)^{\otimes2}\left(F\right)\right]=\mathbb{E}\left[\eta_{0}^{\otimes2}C_{\epsilon_{0}}Q_{1}^{\otimes2}C_{\epsilon_{1}}\cdots Q_{n}^{\otimes2}C_{\epsilon_{n}}(F)\right]\label{eq:cerou_prop}
\end{equation}
and in particular for $F=1\otimes1$,
\begin{equation}
\mathbb{E}\left[\gamma_{n}^{N}(1)^{2}\right]=\mathbb{E}\left[\eta_{0}^{\otimes2}C_{\epsilon_{0}}Q_{1}^{\otimes2}C_{\epsilon_{1}}\cdots C_{\epsilon_{n-1}}Q_{n}^{\otimes2}(1\otimes1)\right]\label{eq:gamma-squared-representation}
\end{equation}
where $C_{1}:=C$, $C_{0}:=Id$ and $\left\{ \epsilon_{n}\right\} _{n\geq0}$
is a sequence of i.i.d., $\left\{ 0,1\right\} $-valued random variables
with distribution
\[
\mathbb{\mathbb{P}}(\epsilon_{n}=1)=1-\mathbb{P}(\epsilon_{n}=0)=\frac{1}{N}.
\]

\end{prop}

\begin{prop}
\label{prop:Cerou-linearity} \citep[Corollary 1.5]{CerouDelMoral2011}
If for each $p\geq0$ there exists a finite constant $c_{p}$ such
that
\begin{equation}
\sup_{n\geq p}\sup_{(x,y)\in E^{2}}\frac{Q_{p,n}(1)(x)}{Q_{p,n}(1)(y)}\leq c_{p},\label{eq:Q_np_condition}
\end{equation}
then for any $n\geq0$,
\[
N>\sum_{s=0}^{n}c_{s}\quad\Rightarrow\quad\mathbb{E}\left[\left(\frac{\gamma_{n}^{N}(1)}{\gamma_{n}(1)}-1\right)^{2}\right]\leq\frac{4}{N}\sum_{s=0}^{n}c_{s}.
\]
 \end{prop}
\begin{rem}
If for each $n\geq0$
\begin{equation}
\delta_{n}:=\sup_{(x,y)\in E^{2}}\frac{G_{n}\left(x\right)}{G_{n}(y)}<\infty\quad\text{and}\quad M_{n,n+m}(x,\cdot)\leq\beta_{n}^{(m)}M_{n,n+m}(y,\cdot),\quad\forall(x,y)\in E^{2}\label{eq:regularity-condition-1}
\end{equation}
for some constants $m\geq1$, $\beta_{n}^{(m)}\in\left[1,\infty\right[$,
then (\ref{eq:Q_np_condition}) is satisfied with $c_{p}=\beta_{p}^{(m)}\prod_{p\leq q<p+m}\delta_{q}$.
For a proof see e.g. \citep[Lemma 1.5]{CerouDelMoral2011}. We note
that the statement of \citep[Corollary 1.5]{CerouDelMoral2011} is
written in terms of the condition (\ref{eq:regularity-condition-1}),
but the proof of \citep[Corollary 1.5]{CerouDelMoral2011} actually
uses (\ref{eq:Q_np_condition}).
\end{rem}

\subsection{The pairs particle system}
\label{sub:The-Pairs-particle-system}

In order to derive the Pairs algorithm, our first step
is to obtain in Proposition \ref{prop:Pairs-theory-1} below an alternative
representation of the formula on the right of (\ref{eq:cerou_prop}).
Define for each $n\geq1$, the kernels,
\[
\mathbf{Q}_{n}^{\left(N\right)}(x,dy):=\frac{1}{N}Q_{n}(\check{x},d\check{y})Q_{n}(\check{x},d\hat{y})+\left(1-\frac{1}{N}\right)Q_{n}(\check{x},d\check{y})Q_{n}(\hat{x},d\hat{y}),
\]
with $y=(\check{y},\hat{y})\in E^{2}$, $x=(\check{x},\hat{x})\in E^{2}$
when $n\geq2$ and $x=(\check{x},\hat{x})\in E_{0}^{2}$ when $n=1$.
Similarly to $Q_{p,n}$ we write for $p<n$, $\mathbf{Q}_{p,n}^{\left(N\right)}:=\mathbf{Q}_{p+1}^{\left(N\right)}\cdots\mathbf{Q}_{n}^{\left(N\right)}$
and $\mathbf{Q}_{n,n}^{\left(N\right)}:=Id$. Note that we can equivalently
write $\mathbf{Q}_{n}^{\left(N\right)}$ using the previously defined
coalescence operator $C$ as:
\[
\mathbf{Q}_{n}^{\left(N\right)}=\frac{1}{N}CQ_{n}^{\otimes2}+\left(1-\frac{1}{N}\right)Q_{n}^{\otimes2}.
\]

\begin{prop}
\label{prop:Pairs-theory-1} For any $n\geq1$, $N\geq2$, and $F\in\mathcal{B}_{b}(E\times E)$,
\begin{equation}
\mathbb{E}\left[(\gamma_{n}^{N})^{\otimes2}(F)\right]=\eta_{0}^{\otimes2}\mathbf{Q}_{0,n}^{(N)}(F_{N}),\label{eq:Q_bold_identity}
\end{equation}
where $F_{N}:=N^{-1}CF+(1-1/N)F$, and in the particular case $F=1\otimes1$,
\begin{equation}
\mathbb{E}\left[\gamma_{n}^{N}(1)^{2}\right]=\eta_{0}^{\otimes2}\mathbf{Q}_{0,n}^{(N)}(1\otimes1).\label{eq:E_gamma_N=00003DQ}
\end{equation}
\end{prop}
\begin{proof}
Starting from the identity of Proposition \ref{prop:Cerou3.4-1}, namely equation
(\ref{eq:cerou_prop}), we have
\begin{eqnarray*}
 &  & \mathbb{E}\left[(\gamma_{n}^{N})^{\otimes2}(F)\right]\\
 & = & \sum_{\epsilon_{0:n}\in\{0,1\}^{n+1}}\eta_{0}^{\otimes2}C_{\epsilon_{0}}Q_{1}^{\otimes2}C_{\epsilon_{1}}\cdots Q_{n}^{\otimes2}C_{\epsilon_{n}}(F)\prod_{p=0}^{n}\left(1-\frac{1}{N}\right)^{\mathbb{I}[\epsilon_{p}=0]}\left(\frac{1}{N}\right)^{\mathbb{I}[\epsilon_{p}=1]}\\
 & = & \sum_{\epsilon_{0:n-1}\in\{0,1\}^{n}}\int_{E_{0}^{2}\times E^{2n}}F_{N}(x_{n})\eta_{0}^{\otimes2}(dx_{0})\prod_{p=1}^{n}\left(C_{\epsilon_{p-1}}Q_{p}^{\otimes2}\right)(x_{p-1},dx_{p})\\
 &  & \left(1-\frac{1}{N}\right)^{\mathbb{I}[\epsilon_{p-1}=0]}\left(\frac{1}{N}\right)^{\mathbb{I}[\epsilon_{p-1}=1]}\\
 & = & \int_{E_{0}^{2}\times E^{2n}}F_{N}(x_{n})\eta_{0}^{\otimes2}(dx_{0})\prod_{p=1}^{n}\mathbf{Q}_{p}^{(N)}(x_{p-1},dx_{p})\\
 & = & \eta_{0}^{\otimes2}\mathbf{Q}_{0,n}^{(N)}(F_{N}),
\end{eqnarray*}
which establishes (\ref{eq:Q_bold_identity}). For (\ref{eq:E_gamma_N=00003DQ}),
note $C(1\otimes1)=1\otimes1$ and $(\gamma_{n}^{N})^{\otimes2}(1\otimes1)=\gamma_{n}^{N}(1)^{2}$.
\end{proof}
Throughout the remainder of this section $N\geq1$ is fixed. Similarly
to (\ref{eq:M_and_G_from_Q}), we now associate with $\left(\mathbf{Q}_{n}^{\left(N\right)}\right)_{)n\geq1}$
collections of Markov kernels $\left(\mathbf{M}_{n}^{(N)}\right)_{n\geq1}$
and positive functions $\left(\mathbf{G}_{n}^{(N)}\right)_{n\geq0}$
, given for $x=(\check{x},\hat{x})\in E^{2}$,
\begin{eqnarray}
\mathbf{G}_{n-1}^{(N)}(x) & := & \mathbf{Q}_{n}^{(N)}(x,E\times E)=\frac{1}{N}G_{n-1}(\check{x})^{2}+\left(1-\frac{1}{N}\right)G_{n-1}(\check{x})G_{n-1}(\hat{x}),\label{eq:bold_G_denf}\\
\mathbf{M}_{n}^{(N)}(x,dy) & := & \frac{\mathbf{Q}_{n}^{(N)}(x,dy)}{\mathbf{Q}_{n}^{(N)}(x,E\times E)}=\dfrac{\mathbf{Q}_{n}^{(N)}(x,dy)}{\int_{E\times E}\mathbf{Q}_{n}^{(N)}\left(x,dz\right)}=\dfrac{\mathbf{Q}_{n}^{(N)}(x,dy)}{\mathbf{G}_{n-1}^{(N)}(x)}\nonumber \\
 & = & p_{n-1}\left(\check{x},\hat{x}\right)M_{n}\left(\check{x},d\check{y}\right)M_{n}\left(\check{x},d\hat{y}\right)\label{eq:bold_M_defn}\\
 & + & (1-p_{n-1}\left(\check{x},\hat{x}\right))M_{n}\left(\check{x},d\check{y}\right)M_{n}\left(\hat{x},d\hat{y}\right),\nonumber
\end{eqnarray}
where
\begin{equation}
p_{n-1}\left(\check{x},\hat{x}\right):=\left[1+(N-1)\dfrac{G_{n-1}(\hat{x})}{G_{n-1}(\check{x})}\right]^{-1}.\label{eq:BernoulliProb-1}
\end{equation}
 Now similarly to (\ref{eq:gamma_n_defn}), define the measures $\left(\Gamma_{n}^{(N)}\right)_{n\geq0}$
and probability measures $\left(H_{n}^{(N)}\right)_{n\geq1}$ according
to $H_{0}^{(N)}:=\Gamma_{0}^{(N)}:=\eta_{0}^{\otimes2}$ and
\begin{equation}
\Gamma_{n}^{(N)}(\cdot):=\eta_{0}^{\otimes2}\mathbf{Q}_{0,n}^{(N)}(\cdot),\quad H_{n}^{(N)}(\cdot):=\frac{\Gamma_{n}^{(N)}(\cdot)}{\Gamma_{n}^{(N)}(E\times E)},\quad n\geq1.\label{eq:big_gamma_defn}
\end{equation}

With these objects, and for some fixed $M\geq1$, we associate a particle
process $\left(\xi_{n}\right)_{n\geq0}$ as follows. The initial configuration
$\xi_{0}=\left\{ \xi_{0}^{1},,\ldots,\xi_{0}^{M}\right\} $ consists
of $M$ i.i.d. pairs, each $\xi_{0}^{i}=(\check{\xi}_{0}^{i},\hat{\xi}_{0}^{i})$
valued in $E_{0}^{2}$ and having distribution $H_{0}^{(N)}=\eta_{0}^{\otimes2}$;
and for $n\geq1$, $\xi_{n}=\left\{ \xi_{n}^{1},,\ldots,\xi_{n}^{M}\right\} $
consists of $M$ pairs, each $\xi_{n}^{i}=(\check{\xi}_{n}^{i},\hat{\xi}_{n}^{i})$
valued in $E^{2}$, with evolution given by:
\begin{eqnarray}
\mathbb{P}(\left.\xi_{n}\in d\xi_{n}\right|\xi_{0},...,\xi_{n-1}): & = & \prod_{i=1}^{M}\frac{\sum_{j=1}^{M}\mathbf{Q}_{n}^{(N)}(\xi_{n-1}^{j},d\xi_{n}^{i})}{\sum_{j=1}^{M}\mathbf{Q}_{n}^{(N)}(\xi_{n-1}^{j},E)}\label{eq:pairs_law-1}\\
 & = & \prod_{i=1}^{M}\frac{\sum_{j=1}^{M}\mathbf{G}_{n-1}^{(N)}(\xi_{n-1}^{j})\mathbf{M}_{n}^{(N)}(\xi_{n-1}^{j},d\xi_{n}^{i})}{\sum_{j=1}^{M}\mathbf{G}_{n-1}^{(N)}(\xi_{n-1}^{j})},\quad n\geq1.\nonumber
\end{eqnarray}
We then introduce the empirical measures
\begin{eqnarray}
 &  & H_{n}^{(N,M)}:=M^{-1}\sum_{i=1}^{M}\delta_{\xi_{n}^{i}},\quad n\geq0,\label{eq:pairs-measures-defn}\\
 &  & \Gamma_{0}^{(N,M)}:=H_{0}^{(N,M)},\quad\Gamma_{n}^{(N,M)}(\cdot):=H_{n}^{(N,M)}(\cdot)\prod_{p=0}^{n-1}H_{p}^{(N,M)}(\mathbf{G}_{p}^{(N)}),\quad n\geq1.\nonumber
\end{eqnarray}

\subsubsection*{Algorithm \ref{alg:PairsAlgo} as an instance of the pairs particle
system.}

Let $\left(\mathsf{X},\mathcal{X}\right)$, $f$, $g$, $\pi_{0}$,
etc. be the ingredients of the HMM, defined in Section \ref{sec:SMC-and-HMM}.
To cast Algorithm \ref{alg:PairsAlgo} as an instance of the pairs
particle system described above, we just make the same choices as
in (\ref{eq:Algo1-as-GPS_1})-(\ref{eq:Algo1-as-GPS_2}). Moreover,
in that situation observe that for $\Xi_{n}^{(N,M)}$ as appearing
in Algorithm \ref{alg:PairsAlgo},
\begin{equation}
\Gamma_{n+1}^{(N,M)}(1\otimes1)\equiv\Xi_{n}^{(N,M)}\label{eq:Gam_equiv_Xi}
\end{equation}

\subsection{Proof of Theorem \ref{thm:Pairs-convergence-and-bound}}
\label{sub:Proof-of-Theorem}

To conclude the paper, we gather together various facts from the preceeding
sections of the appendix and complete the proof of Theorem \ref{thm:Pairs-convergence-and-bound}.
\begin{proof}[Proof of Theorem \ref{thm:Pairs-convergence-and-bound}]
\emph{}Unless stated otherwise, throughout the proof $N\geq2$ is
fixed to an arbitrary value. Comparing (\ref{eq:pairs_law-1}) with
(\ref{eq:generic_law}), we see that the pairs particle system described
in Section \ref{sub:The-Pairs-particle-system} is itself an instance
of the generic particle system described in Section \ref{sub:A-generic-particle};
in place of $E_{0}$, $\eta_{0}$, $E$, $G_{n}$, $M_{n}$ etc. in
the latter take $E_{0}^{2}$, $\eta_{0}^{\otimes2}$, $E^{2}$, $\mathbf{G}_{n}^{(N)}$,
$\mathbf{M}_{n}^{(N)}$ etc. This observation allows us to transfer
the various properties described in Section \ref{sub:A-generic-particle}
over to the pairs particle system, as follows.

Firstly, (\ref{eq:G_bounded})-(\ref{eq:standard-Convergence}) read
in this situation as: if for each $n\geq0$,
\begin{equation}
\sup_{x}\mathbf{G}_{n}^{(N)}(x)<\infty,\label{eq:bold_G_bounded}
\end{equation}
 then for any $F\in\mathcal{B}_{b}(E\times E)$,
\begin{equation}
H_{n}^{(N,M)}(F)\stackrel[M\rightarrow\infty]{a.s.}{\longrightarrow}H_{n}^{(N)}(F),\quad\quad\Gamma_{n}^{(N,M)}(F)\stackrel[M\rightarrow\infty]{a.s.}{\longrightarrow}\Gamma_{n}^{(N)}(F).\label{eq:pairs-convergence}
\end{equation}
 Secondly, the lack-of-bias property (\ref{eq:particle-filter-unbiasedness}),
combined with (\ref{eq:big_gamma_defn}) and (\ref{eq:E_gamma_N=00003DQ}),
reads as:
\begin{equation}
\mathbb{E}\left[\Gamma_{n}^{(N,M)}(1\otimes1)\right]=\Gamma_{n}^{(N)}(1\otimes1)=\eta_{0}^{\otimes2}\mathbf{Q}_{0,n}^{(N)}(1\otimes1)=\mathbb{E}\left[\gamma_{n}^{N}(1)^{2}\right],\quad\forall M\geq1.\label{eq:big_Gamma_id}
\end{equation}
Thirdly, Proposition \ref{prop:Cerou-linearity} reads: if for each
$p\geq0$ there exists a finite constant $\mathbf{c}_{p}$ such that
\begin{equation}
\sup_{n\geq p}\sup_{(x,y)\in E^{4}}\frac{\mathbf{Q}_{p,n}^{(N)}(1)(x)}{\mathbf{Q}_{p,n}^{(N)}(1)(y)}\leq\mathbf{c}_{p},\label{eq:boldQ_np_condition}
\end{equation}
then for any $n\geq0$,
\begin{equation}
M>\sum_{s=0}^{n}\mathbf{c}_{s}\quad\Rightarrow\quad\mathbb{E}\left[\left(\frac{\Gamma_{n}^{(N,M)}(1\otimes1)}{\mathbb{E}\left[\gamma_{n}^{N}(1)^{2}\right]}-1\right)^{2}\right]\leq\frac{4}{M}\sum_{s=0}^{n}\mathbf{c}_{s},\label{eq:Gamma_bound}
\end{equation}
where in writing the l.h.s. of the inequality in (\ref{eq:Gamma_bound}),
the identity $\Gamma_{n}^{(N)}(1\otimes1)=\mathbb{E}\left[\gamma_{n}^{N}(1)^{2}\right]$
from (\ref{eq:big_Gamma_id}) has been applied.

To complete the proof of Theorem \ref{thm:Pairs-convergence-and-bound}
it remains to show that in the setting (\ref{eq:Algo1-as-GPS_1})-(\ref{eq:Algo1-as-GPS_2}),
the conditions (\ref{eq:main_thm_weights_bounded}) and (\ref{eq:main_thm_reg_weights})-(\ref{eq:main_thm_reg_q})
imply respectively (\ref{eq:bold_G_bounded}) and (\ref{eq:boldQ_np_condition})
for suitable constants $\mathbf{c}_{p}$ which do not depend on $N$,
since then re-writting (\ref{eq:pairs-convergence}), (\ref{eq:big_Gamma_id})
and (\ref{eq:Gamma_bound}) using (\ref{eq:gamma_equiv_Z}) and (\ref{eq:Gam_equiv_Xi})
gives the claims of the Theorem.

The condition (\ref{eq:main_thm_weights_bounded}) does indeed imply
(\ref{eq:bold_G_bounded}), since by (\ref{eq:bold_G_denf}), $\sup_{x}\mathbf{G}_{n}^{(N)}(x)=\sup_{x}G_{n}(x)^{2}$
for any $N$. It remains to establish (\ref{eq:boldQ_np_condition}).
We first observe that with $G_{p}$ as in (\ref{eq:Algo1-as-GPS_1})-(\ref{eq:Algo1-as-GPS_2}),
conditions (\ref{eq:main_thm_reg_weights})-(\ref{eq:main_thm_reg_weights-1})
imply that there for each $p\geq0$,
\[
\mathbf{d}_{p}:=\sup_{x,y}\frac{\mathbf{G}_{p}^{(N)}(x)}{\mathbf{G}_{p}^{(N)}(y)}=\sup_{x,y}\frac{G_{p}(x)^{2}}{G_{p}(y)^{2}}\leq\left(\frac{w_{p}^{+}}{w_{p}^{-}}\right)^{2}<+\infty.
\]
Now consider (\ref{eq:boldQ_np_condition}) for some given $p.$ When
$n\leq p+1$,
\[
\frac{\mathbf{Q}_{p,n}^{(N)}(1)(x)}{\mathbf{Q}_{p,n}^{(N)}(1)(y)}\leq\mathbf{d}_{p}.
\]
For $n\geq p+2$, suppose there exist contants $0<\mathbf{k}_{p}^{-}\leq\mathbf{k}_{p}^{+}<+\infty$
independent of $N$, and $\mathbf{m}_{p}^{(N)}\in\mathcal{P}(\mathsf{X}^{2}\times\mathsf{X}^{2})$
such that
\begin{equation}
\mathbf{k}_{p}^{-}\mathbf{m}_{p}^{(N)}(\cdot)\leq\mathbf{Q}_{p,p+2}^{(N)}(x,\cdot)\leq\mathbf{k}_{p}^{+}\mathbf{m}_{p}^{(N)}(\cdot),\quad\forall x.\label{eq:bold_Q_mixing}
\end{equation}
Then
\[
\frac{\mathbf{Q}_{p,n}^{(N)}(1)(x)}{\mathbf{Q}_{p,n}^{(N)}(1)(y)}=\frac{\mathbf{Q}_{p,p+2}^{(N)}\mathbf{Q}_{p+2,n}^{(N)}(1)(x)}{\mathbf{Q}_{p,p+2}^{(N)}\mathbf{Q}_{p+2,n}^{(N)}(1)(y)}\leq\frac{\mathbf{k}_{p}^{+}}{\mathbf{k}_{p}^{-}}\frac{\mathbf{m}_{p}^{(N)}\mathbf{Q}_{p+2,n}^{(N)}(1)}{\mathbf{m}_{p}^{(N)}\mathbf{Q}_{p+2,n}^{(N)}(1)}=\frac{\mathbf{k}_{p}^{+}}{\mathbf{k}_{p}^{-}},
\]
and (\ref{eq:boldQ_np_condition}) would then hold with $\mathbf{c}_{p}:=\mathbf{d}_{p}\vee\frac{\mathbf{k}_{p}^{+}}{\mathbf{k}_{p}^{-}}$.
Thus to complete the proof we shall show that conditions (\ref{eq:main_thm_reg_weights})-(\ref{eq:main_thm_reg_q})
imply (\ref{eq:bold_Q_mixing}). To this end note that:
\begin{eqnarray}
\mathbf{Q}_{p,p+2}^{(N)} & = & \left(1-\frac{1}{N}\right)\left[\frac{1}{N}C+\left(1-\frac{1}{N}\right)Id\right]Q_{p+1}^{\otimes2}Q_{p+2}^{\otimes2}\nonumber \\
 & + & \frac{1}{N}\left[\frac{1}{N}C+\left(1-\frac{1}{N}\right)Id\right]Q_{p+1}^{\otimes2}CQ_{p+2}^{\otimes2},\label{eq:Q^N_p,p+2}
\end{eqnarray}
and with
\[
\mathbf{k}_{p}^{-}:=\left(w_{p}^{-}w_{p+1}^{-}\epsilon_{p+1}^{-}\right)^{2},\quad\quad\mathbf{k}_{p}^{+}:=\left(w_{p}^{+}w_{p+1}^{+}\epsilon_{p+1}^{+}\right)^{2},
\]
for all $x=(x_{1},x_{2})$,
\begin{equation}
\mathbf{k}_{p}^{-}\mu_{p+1}^{\otimes2}(dy_{1})q_{p+2}^{\otimes2}(y_{1},dy_{2})\leq Q_{p+1}^{\otimes2}Q_{p+2}^{\otimes2}(x,dy)\leq\mathbf{k}_{p}^{+}\mu_{p+1}^{\otimes2}(dy_{1})q_{p+2}^{\otimes2}(y_{1},dy_{2})\label{eq:Q_two_step_bounds}
\end{equation}
 and
\begin{eqnarray}
\mathbf{k}_{p}^{-}\int_{\mathsf{X}}\mu_{p+1}(dz)\delta_{z}^{\otimes2}(dy_{1})q_{p+2}^{\otimes2}(y_{1},dy_{2}) & \leq & Q_{p+1}^{\otimes2}CQ_{p+2}^{\otimes2}(x,dy)\leq\label{eq:Q_two_step_C_bounds}\\
 & \leq & \mathbf{k}_{p}^{+}\int_{\mathsf{X}}\mu_{p+1}(dz)\delta_{z}^{\otimes2}(dy_{1})q_{p+2}^{\otimes2}(y_{1},dy_{2}),\nonumber
\end{eqnarray}
where $\delta_{z}(\cdot)$ is the Dirac measure on $\mathsf{X}$ located
at $z$, and $dy=dy_{1}dy_{2}$ is to be understood as the infinitesimal
neighbourhood of $y=(y_{1},y_{2})\in\mathsf{X}^{2}\times\mathsf{X}^{2}$.
Combining (\ref{eq:Q^N_p,p+2})-(\ref{eq:Q_two_step_C_bounds}) we
find that (\ref{eq:bold_Q_mixing}) holds with
\[
\mathbf{m}_{p}^{(N)}(dy):=\left(1-\frac{1}{N}\right)\mu_{p+1}^{\otimes2}(dy_{1})q_{p+2}^{\otimes2}(y_{1},dy_{2})+\frac{1}{N}\int_{\mathsf{X}}\mu_{p+1}(dz)\delta_{z}^{\otimes2}(dy_{1})q_{p+2}^{\otimes2}(y_{1},dy_{2}).
\]

\end{proof}

\begin{acknowledgements}
SK would like to thank the University of Bristol for providing him
with University of Bristol Postgraduate Scholarship during the time
this research was done.
\end{acknowledgements}

\bibliographystyle{spbasic}      
\bibliography{Pairs_algo_and_Marginal_likelihood_Springer_style}   

%
%

\end{document}